\documentclass[letterpaper,twocolumn,10pt]{article}
\usepackage{usenix}

\usepackage{tikz}
\usepackage{amsmath}
\usepackage{amsthm}
\newtheorem{theorem}{Theorem}[section]

\usepackage[linesnumbered,ruled,vlined]{algorithm2e}
\usepackage{algorithmic}

\usepackage{multirow}
\usepackage{caption}
\usepackage{booktabs}

\usepackage{amssymb}
\usepackage{makecell}
\usepackage{wasysym}
\usepackage{enumerate}
\usepackage{graphicx}
\usepackage{threeparttable}
\usepackage{tcolorbox}
\usepackage{url}
\usepackage[T1]{fontenc}
\usepackage{xcolor}

\newcommand{\MyScheme}{\ensuremath{\mathtt{Amulet}}\xspace}

\begin{document}

\date{}

\title{\Large \bf Amulet: Fast TEE-Shielded Inference for On-Device Model Protection}

\author{
{\rm Zikai Mao{$^1$}, Lingchen Zhao{$^1$}\thanks{Corresponding author.}, Lei Xu{$^2$}, Wentao Dong{$^3$}, Shenyi Zhang{$^1$},}\\
{\rm Cong Wang{$^3$}, and Qian Wang{$^1$}}\\
{$^1$} Key Laboratory of Aerospace Information Security and Trusted Computing, Ministry of Education,\\ School of Cyber Science and Engineering, Wuhan University,\\
{$^2$} School of Mathematics and Statistics, Nanjing University of Science and Technology, \\
{$^3$} City University of Hong Kong
}
\maketitle

\begin{abstract}
On-device machine learning (ML) introduces new security concerns about model privacy. Storing valuable trained ML models on user devices exposes them to potential extraction by adversaries. The current mainstream solution for on-device model protection is storing the weights and conducting inference within Trusted Execution Environments (TEEs). However, due to limited trusted memory that cannot accommodate the whole model, most existing approaches employ a partitioning strategy, dividing a model into multiple slices that are loaded into the TEE sequentially. This frequent interaction between untrusted and trusted worlds dramatically increases inference latency, sometimes by orders of magnitude.

In this paper, we propose Amulet, a fast TEE-shielded on-device inference framework for ML model protection. Amulet incorporates a suite of obfuscation methods specifically designed for common neural network architectures such as convolutional layers, attention blocks, and non-linear activation functions. After obfuscation by the TEE, the entire transformed model can be securely stored in untrusted memory, allowing the inference process to execute directly in untrusted memory with GPU acceleration. For each inference request, only two rounds of minimal-overhead interaction between untrusted and trusted memory are required to process input samples and output results. We also provide theoretical proof from an information-theoretic perspective that the obfuscated model does not leak information about the original weights.
We comprehensively evaluated Amulet using diverse model architectures ranging from ResNet-18 to GPT-2. Our approach incurs inference latency only 2.8-4.8$\times$ that of unprotected models with negligible accuracy loss, achieving an 8-9$\times$ speedup over baseline methods that execute inference entirely within TEEs, and performing approximately 2.2$\times$ faster than the state-of-the-art obfuscation-based method.
\end{abstract}

\section{Introduction}

With the rapid advancement of computational capabilities in mobile devices, on-device machine learning (ML) has emerged as a crucial paradigm for privacy-sensitive and latency-critical AI applications~\cite{zhao2022survey}. Deep learning models designed for terminal devices such as smartphones, embedded systems, and laptops have been widely utilized in fields like autonomous driving, image recognition, and chatbot applications, enabling real-time decision-making~\cite{li2024personal}. By processing data locally, on-device ML not only minimizes latency but also enhances user privacy by avoiding the transmission of sensitive information to the cloud. Additionally, it ensures network independence, making it adaptable to various environments without requiring a stable internet connection.

However, the deployment of ML models on untrusted devices raises significant concerns regarding intellectual property risks. Attackers may exploit vulnerabilities in these devices to steal proprietary models, posing a critical challenge to the security and privacy of on-device ML systems~\cite{nayan2024sok}. Since models are often stored in plaintext on local devices, attackers can potentially access and extract the model weights and architecture, thereby compromising the proprietary knowledge embedded in the model~\cite{sun2021mind}. This poses a serious threat to the ownership and commercial value of these models.

Leveraging Trusted Execution Environments (TEEs) to implement on-device inference represents a mainstream approach to protect models~\cite{shen2022soter, sun2023shadownet,zhang2024no,sun2024tsqp,wang2025arrow}. TEEs establish encrypted and isolated hardware enclaves that provide trusted execution environments, ensuring program integrity and data confidentiality~\cite{costan2016intel, inteltdx, pinto2019demystifying, sev2020strengthening}. By storing models within TEEs and performing inference internally, this approach can, under ideal circumstances, address intellectual property concerns. Since most personal devices today are equipped with TEE functionality~\cite{appleTEE, iotTEE}, such approaches seem to be practical for deployment. However, based on current findings, this technology still faces challenges in computational efficiency and security when applied in practical scenarios. We summarize the current limitations of TEE-shielded on-device inference as follows:

\begin{enumerate}[1)]
\item 
\textbf{Efficiency Bottlenecks.} 
Owing to the limited trusted memory space available in TEEs, which may be insufficient to store models with large parameter counts, most existing methods adopt a model-splitting strategy, where only portions of the model are loaded into the TEE at a time during inference. This approach introduces a significant problem: it necessitates extensive interaction between the TEE and external untrusted memory, resulting in substantial inter-memory communication overhead and numerous public-key cryptographic operations such as encryption, decryption, and signatures. These overheads scale with model size and preclude the full utilization of GPU or NPU acceleration, further degrading computational efficiency. For instance, Slalom already experiences a 14$\times$ performance degradation compared to unprotected inference for ResNet-18~\cite{tramer2018slalom}. With the adoption of the next-generation TEEs such as SGX2 and NVIDIA GPUs~\cite{emily2023h100} that support larger secure memory, this limitation can be mitigated in cloud-based secure inference tasks. Unfortunately, most of them currently target only high-performance processors for cloud infrastructures and do not support consumer-grade devices such as smartphones and laptops, which inherently have limited memory and even smaller available secure memory. This leads to the question:

\begin{tcolorbox}[size=small]
    \textbf{RQ1:} Is it possible to achieve minimal interaction with the TEE while preserving model privacy?
\end{tcolorbox}

\item 
\textbf{Security Vulnerabilities.} Given that completing the entire inference process within TEEs may incur unacceptable overhead, recent works have proposed solutions to improve efficiency by reducing TEE usage. However, these schemes may suffer from two types of vulnerabilities. The first is direct leakage of model information. Schemes like DarkneTZ~\cite{mo2020darknetz}, which only place part of the model structure inside the TEE, provide protection for only some layers of the model, leaving the remaining layers completely unprotected. The second is indirect leakage of model information. Approaches like ShadowNet~\cite{sun2023shadownet} attempt to protect the model by obfuscating weights, but attackers can still infer information by analyzing the distribution patterns of these obfuscated weights or the access patterns of secure memory, as demonstrated in~\cite{zhang2024no, yuan2024hypertheft}. This leads to the question:
\begin{tcolorbox}[size=small]
    \textbf{RQ2:} Is it possible to provide comprehensive protection for model privacy while maintaining inference efficiency?
\end{tcolorbox}
\end{enumerate}

\subsection{Our Contribution}

In this paper, we present \MyScheme{}, a novel on-device secure inference framework that addresses the aforementioned challenges. \MyScheme{} adopts a hybrid approach combining TEEs and weight obfuscation techniques, achieving constant-round interaction with TEEs (requiring only two interactions per inference), supporting arbitrary external accelerators such as GPUs, and providing provable security for the entire model. We summarize our contributions as follows:

\noindent \textbf{TEE minimization \& GPU maximization.}
To avoid frequent memory swapping caused by trusted memory limitations during the inference phase, we propose a novel approach that exchanges only data rather than model parameters, differing from previous model splitting strategies. The TEE completes the obfuscation of model parameters before inference begins. When a user initiates an inference request, only one interaction is needed to transfer the input to the TEE and return a transformation to untrusted memory. The transformed input is then used for inference on the obfuscated model, producing an obfuscated inference result. Finally, this result is transmitted back to the TEE for restoration, obtaining the correct inference result.

In this process, the model obfuscation is input-independent and can be completed during system idle time as a preprocessing stage, without affecting inference efficiency. Untrusted memory only needs to interact with the TEE twice, transferring minimal information about inputs and outputs to complete the inference, resulting in extremely low overhead that does not increase with model size. Furthermore, the inference process on the obfuscated model involves operations similar to normal inference, primarily based on matrix multiplication, allowing utilization of accelerators such as GPUs to enhance computational efficiency.

\noindent
\textbf{Full layer protection.}
To ensure that the obfuscated model maintains high computational efficiency and strong security guarantees, we design a series of obfuscation methods for different neural network structures, such as convolution, batch normalization, Rectified Linear Unit (ReLU) function, and multi-head attention components. We categorize neural network components into two types: linear layers and non-linear layers.
Linear layers form the core of neural networks, comprising numerous model parameters and primarily implemented through matrix operations. We protect these by using bidirectional masking to obfuscate them into random matrices and performing inference directly on these obfuscated matrices.
Non-linear layers, implemented through element-wise operations rather than matrix computations, cannot directly utilize the linear layer obfuscation scheme. For these, we design an absorb-shuffle-squeeze process, which first absorbs the non-linear layer inputs into a matrix, applies linear layer obfuscation methods to complete the computation, and then squeezes the results back to their original form.
Based on these specially designed obfuscation methods for different components, we can construct most mainstream models currently in use, including small-scale image recognition models as well as large language models.

\begin{table*}[!t]
    \caption{Comparative Analysis of Existing TEE-shielded DNN Partition Frameworks: Support Levels for Key Characteristics (\CIRCLE~ = Full support, \LEFTcircle~ = Partial support, \Circle~ = No support). ${n}$ represents the number of layers placed in the TEE.} 
    
    \label{tab:model-characteristics}
    \centering
    \resizebox{0.98\textwidth}{!}{
        \begin{tabular}{@{}l|ccccccc@{}}
        \toprule
        \textbf{System}& \makecell{\textbf{Model} \\ \textbf{Privacy}} & 
        \makecell{\textbf{High} \\ \textbf{Accuracy}} & 
        \makecell{\textbf{GPU} \\ \textbf{Acceleration}} &
        \makecell{\textbf{Support} \\ \textbf{Transformers}} &
        \makecell{\textbf{Protection} \\ \textbf{Level}} & 
        \makecell{\textbf{Portions in} \\ \textbf{TEE}} &
        \makecell{\textbf{Interactions} \\ \textbf{with TEE}} \\
        \midrule
        MLCapsule~\cite{hanzlik2021mlcapsule}       & \CIRCLE & \Circle & \Circle & \Circle & \CIRCLE & All layers & $O(1)$ \\
        Serdab~\cite{elgamal2020serdab}          & \Circle & \CIRCLE & \CIRCLE & \Circle & \LEFTcircle & Shallow layers & $O(1)$ \\
        DarkneTZ~\cite{mo2020darknetz}        & \Circle & \CIRCLE & \CIRCLE & \Circle & \LEFTcircle & Deep layers & $O(1)$ \\
        AegisDNN~\cite{xiang2021aegisdnn}        & \Circle & \CIRCLE & \CIRCLE & \Circle & \LEFTcircle & Intermediate layers & $O(1)$ \\
        Slalom~\cite{tramer2018slalom}          & \Circle & \CIRCLE & \CIRCLE & \CIRCLE & \CIRCLE & Non-linear layers & $O(n)$ \\
        SOTER~\cite{shen2022soter}           & \Circle & \CIRCLE & \CIRCLE & \CIRCLE & \LEFTcircle & Intermediate layers & $O(n)$ \\
        ShadowNet~\cite{sun2023shadownet}       & \Circle & \CIRCLE & \CIRCLE & \Circle & \CIRCLE & Non-linear layers & $O(n)$ \\
        MirrorNet~\cite{liu2023mirrornet}       & \Circle & \CIRCLE & \CIRCLE & \Circle & \LEFTcircle & Layer slices & $O(n)$ \\
        TEESlice~\cite{zhang2024no}        & \CIRCLE & \CIRCLE & \CIRCLE & \CIRCLE & \CIRCLE & Layer slices & $O(n)$ \\
        GroupCover~\cite{zhanggroupcover}      & \CIRCLE & \CIRCLE & \CIRCLE & \Circle & \CIRCLE & Non-linear layers & $O(n)$ \\
        TSQP~\cite{sun2024tsqp}            & \CIRCLE & \CIRCLE & \CIRCLE & \CIRCLE & \CIRCLE & Non-linear layers & $O(n)$ \\
        ArrowCloak~\cite{wang2025arrow} & \CIRCLE & \CIRCLE & \CIRCLE & \CIRCLE & \CIRCLE & Non-linear layers & $O(n)$ \\
        \textbf{\MyScheme}        & \CIRCLE & \CIRCLE & \CIRCLE & \CIRCLE & \CIRCLE & \textbf{Inputs \& Outputs} & $\boldsymbol{O(1)}$ \\
        \bottomrule
        \end{tabular}
        }
\end{table*}

\noindent
\textbf{Implementation \& validation.}
We have implemented and thoroughly evaluated the efficiency and availability of \MyScheme{} across popular models of varying scales, such as VGG, ResNet, BERT, and GPT-2. In terms of efficiency, \MyScheme{} achieves inference times only 2.8-4.8$\times$ that of unprotected original models, representing approximately an 8-9$\times$ improvement over baseline methods that execute the entire inference process within TEEs. Regarding accuracy, since \MyScheme{} preserve both the model structure and weight values, it produces results that are nearly identical to those of the original model, with deviations limited to minor computational errors only arising from floating-point precision.

\section{Related Works}
Generally, the existing solutions about TEE-Shielded DNN Partition inferences can be categorized into two types: shielding full layers, shielding a part of layers.

\noindent
\textbf{Shielding Full Layers of the Model.} Storing and executing all model layers within the TEE is the most straightforward solution. This method was first introduced in MLCapsule~\cite{hanzlik2021mlcapsule}, where the idea of encapsulating every layer of the model in the enclave was proposed to provide protection for sensitive data. As model sizes continue to grow, the limited trusted memory of TEEs deployed in edge devices has become the primary bottleneck constraining the efficiency of these methods.

To improve the efficiency of TEE-shielded inference, some works have proposed shielding only non-linear layers while obfuscating the others.
Slalom~\cite{tramer2018slalom} adopted an approach of outsourcing linear layers to a GPU for accelerated computation after applying one-time-pad encryption, while non-linear layers are securely computed and protected within the TEE. This method achieved significant performance improvements compared to MLCapsule by reducing the interaction between untrusted and trusted memory while leveraging GPU acceleration, inspiring numerous subsequent advancements. 
For instance, SOTER~\cite{shen2022soter} built upon this approach by using selected scalar obfuscation on the outsourced linear layers before returning them to the enclave for further computation. 
Similarly, ShadowNet~\cite{sun2023shadownet} extended the idea by performing linear transformations on the convolutional layer weights before offloading them to the GPU, while ensuring that the ReLU layers remain protected within the TEE.

However, these schemes were subsequently found to have serious security flaws. GroupCover~\cite{zhanggroupcover} presented a detailed attack scheme based on similarity search, designing reverse clustering methods for convolutional kernels. TEESlice~\cite{zhang2024no}, an advanced iteration of the vertical splitting approach, introduced feature encryption for the public linear layers to mitigate security risks. TSQP~\cite{sun2024tsqp} further expanded this framework to quantized neural networks. ArrowCloak~\cite{wang2025arrow} proposed a obfuscation scheme against direction similarity attacks, albeit at a cost of roughly 1.5× slower inference latency than TSQP. Still, all the obfuscation methods require frequent enclave-external interaction, introducing significant and often unnecessary overheads.

\noindent
\textbf{Shielding Parts of Layers.} Due to the high overhead of protecting all layers, some approaches propose enhancing efficiency by weakening security protections and safeguarding only partial model information. These approaches can be categorized into two strategies: horizontal partitioning and vertical splitting of the model.

The horizontal partitioning approach is relatively straightforward, involving placing only certain layers within the TEE for inference while leaving other layers unprotected. This includes schemes that protect shallow layers~\cite{elgamal2020serdab}, intermediate layers~\cite{xiang2021aegisdnn}, and deep layers~\cite{mo2020darknetz}. The drawback of such approaches is evident: they directly expose information from unprotected layers.
Vertical splitting methods, in contrast, partition the original model into two components of different sizes. For example, the smaller component executes within the TEE while the larger component runs on a GPU in untrusted memory~\cite{liu2023mirrornet, liu2024tbnet}. However, these approaches present significant implementation challenges as they require careful design of the partitioned architecture and additional training processes. Furthermore, since the larger unprotected component still contains partial knowledge learned by the original model, the lack of protection for this part may still pose privacy risks.

\section{Overview}

\subsection{Threat Model}
In \MyScheme, we consider a straightforward scenario where only one entity exists, i.e., the user's device, such as a laptop or a smartphone. The device is equipped with the TEE. A pre-trained machine learning model, obtained from the service provider, is deployed on the device. The goal of \MyScheme is to achieve efficient and accurate on-device ML inference without revealing the original model to the user, leveraging the capabilities of the TEE.
Below are our assumptions regarding the capabilities of the TEE and the adversary:

\noindent
\textbf{TEE}. The TEE has access to a limited space of trusted memory. We assume that data stored in this trusted memory and information related to operations executed within it remain inaccessible to the adversary. From the users' perspective, during each TEE invocation, they can only observe the inputs provided to the TEE and the corresponding outputs returned by the TEE. Additionally, the TEE supports remote attestation functionality, enabling authentication that the model deployed on the device was downloaded from a trusted service provider. Note that since this work focuses on preserving the privacy of the model rather than the data, and the data is owned by the device owner, with inference processes completed entirely locally without interaction with other entities, we do not address the privacy of the input data.

\noindent
\textbf{Adversary}. We consider the adversary, i.e., a malicious user, who may attempt to extract information from the trained model. The adversary has control over all parts of the device, including the operating system, hardware components such as GPUs, and external memory, except the TEE and its trusted memory. We assume an active adversary capable of initiating \emph{an arbitrary number of ML inference queries with arbitrary inputs} against the protected model. This threat model encompasses various attack vectors, such as known-plaintext attacks and chosen-plaintext attacks.

Furthermore, adversaries may extract partial model or training data information through pure black-box attacks~\cite{orekondy2019knockoff}. These attacks are challenging to defend against as they only require the model to function normally. Consistent with prior works~\cite{sun2024tsqp, zhang2024no}, our goal is not to defend against black-box attacks but to ensure attackers cannot gain knowledge beyond what pure black-box access permits.

\begin{figure}[ht]
    \centering
    \includegraphics[width=\linewidth]{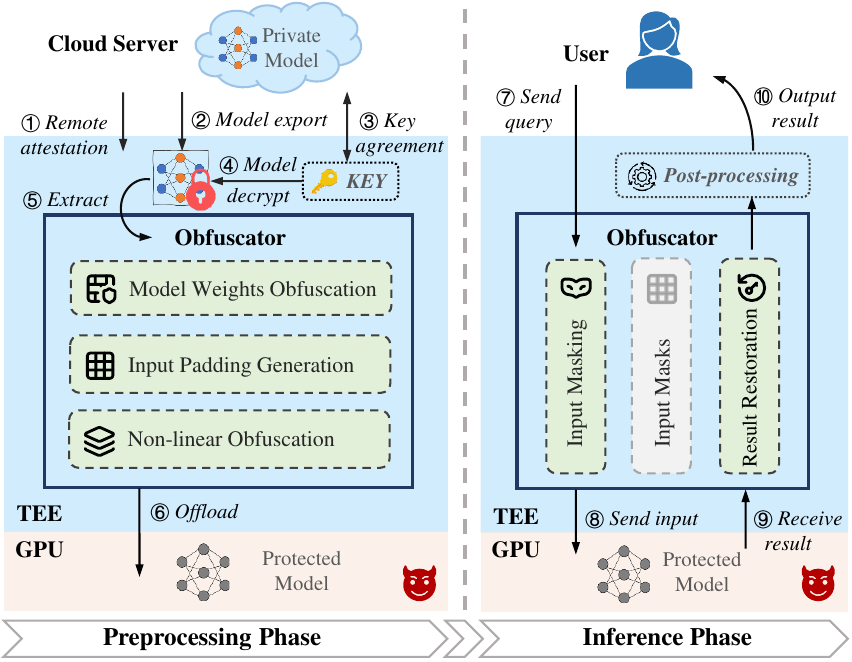} 
    \caption{Architectural Overview of \MyScheme}
    \label{fig:workflow}
\end{figure}

\subsection{Workflow}

The workflow of \MyScheme is shown in Fig.~\ref{fig:workflow}, which consists of two phases: the \textit{preprocessing phase} (\textcircled{\footnotesize{1}} - \textcircled{\footnotesize{6}}) and the \textit{inference phase} (\textcircled{\footnotesize{7}} - \textcircled{\footnotesize{10}}). Generally, the user's end device downloads an encrypted DNN model from the service provider and utilizes the remote attestation function of TEE to securely import the decryption key directly into the TEE. Subsequently, the entire inference process will be conducted locally with the help of TEE, without the need for further interaction with the server. Now we introduce the details as follows:

\noindent
\textbf{Preprocessing phase.}
The goal of the input-independent preprocessing phase is to securely download the model from the server to the TEE on the user's device and perform the obfuscation of the original model entirely within the TEE.
To begin with, the TEE initiates a remote attestation process (\textcircled{\footnotesize{1}}) with the cloud server. For instance, in the case of Intel SGX, it involves leveraging the Intel EPID Attestation technology. This step verifies the integrity and trustworthiness of the SGX enclave running on the device by generating a secure, signed report transmitted to the cloud server. The report serves as a guarantee that the device has not been compromised and is capable of executing trusted operations securely. Upon receiving the report, the cloud server validates the attestation and returns an authentication token to the enclave, authorizing the enclave to process the sensitive data.

Once the remote attestation is successfully completed, the user retrieves and downloads the encrypted model weights from the cloud server (\textcircled{\footnotesize{2}}). Using a securely established key, which is derived from the attestation process (\textcircled{\footnotesize{3}}), the model is decrypted within the TEE(\textcircled{\footnotesize{4}}). This ensures that the model weights remain confidential throughout the entire process, as the decryption occurs exclusively within the TEE and is never exposed to the external operating system or potential attackers. After the decryption, the TEE performs model obfuscation (\textcircled{\footnotesize{5}}), the core step that ensures the sensitive model weights remain protected during the inference phase. Finally, it offloads the obfuscated weight to untrusted memory (\textcircled{\footnotesize{6}}).

\noindent
\textbf{Inference phase.}
The goal of the inference phase is to enable the untrusted GPU to perform the inference process using the obfuscated model weights.
When the user submits a request for inference (\textcircled{\footnotesize{7}}), the input data are first converted into tensors within the untrusted memory. These tensors are then transferred into the TEE, where they are masked by the Obfuscator that was constructed during the preprocessing phase. The TEE subsequently outputs the masked request tensors back to the untrusted memory (\textcircled{\footnotesize{8}}), where they are processed by the obfuscated model and untrusted GPU. Once the inference is completed in the normal execution environment, the masked results are sent back to the TEE for unmasking (\textcircled{\footnotesize{9}}). Finally, the TEE reveals the results to the user, ensuring the security of the model and intermediate results security by providing only the final output (e.g., classification labels) without disclosing the raw confidence scores generated by the model (\textcircled{\footnotesize{10}}).

\noindent
\emph{Remark.} The core innovation of \MyScheme lies in the design of the Obfuscator. The implementation details of downloading an encrypted model from the server, as well as the interaction between the TEE and the untrusted environments, are similar to those in most related works~\cite{sun2024tsqp, zhang2024no, tramer2018slalom}. These steps adhere to the principle that only encrypted or obfuscated data should be stored and processed in untrusted memory. Therefore, we omit the details of these steps. Interested readers are encouraged to refer to \cite{costan2016intel} for further insights.
In addition, the obfuscation step (\textcircled{\footnotesize{5}}) can also be completed on the server side as long as both the server and client know the masks used. We chose to complete it on the client side considering that devices may not always be connected to the server, and this overhead is incurred only once.

\section{Design of Amulet}

In this section, we present the detailed design of \MyScheme{}. 
The inference process of a neural network can be formally represented as an operator sequence $\textbf{NN}:= \{\mathbf{L}_1, ..., \mathbf{L}_n\}$, where $n$ denotes the number of layers in the model. Generally, the operations of these layers can be categorized into two types: \textit{linear layers}, such as dense layers and convolutional layers, and \textit{non-linear layers}, such as the ReLU function and the Gaussian Error Linear Unit (GELU) function. Our obfuscation scheme follows this architecture, sequentially obfuscating and computing each layer. Specifically, we have designed distinct obfuscation patterns for each type of operation, tailored to the specific characteristics of both linear and non-linear transformations.

\subsection{Intuition}
Our solution begins with the relatively straightforward case of linear operations. The linear operations in DNNs can be represented as matrix multiplications. To obfuscate the weights of the linear layers, i.e.,  one of the matrices involved in the matrix multiplication, we employ the following encoding method to implement a linear transformation: $\tilde{X} = P \cdot X \cdot Q$,
where $P$ and $Q$ are matrices whose elements are uniformly and randomly sampled from a predefined distribution $U(-1, 1)$. This approach does not require costly cryptographic operations and can leverage GPU acceleration capabilities due to its implementation through matrix multiplication. If we can utilize these random matrices to replace the parameters of the original model's linear layers and complete the inference process, we can effectively protect the model weights during inference.

Non-linear layers might initially appear exempt from the need for obfuscation because they lack trainable weights like linear layers. However, this assumption overlooks a critical vulnerability that disclosing the outputs of non-linear layers to an adversary could enable the inference of their inputs, which correspond to the un-obfuscated outputs of preceding linear layers. Armed with this information, the adversary could formulate linear equations based on query responses and solve for the weights and biases of the linear layers, thereby reconstructing the entire linear layers with high fidelity, as demonstrated in \cite{carlini2020cryptanalytic}. Therefore, obfuscating non-linear layers is also essential.

Unlike linear layers, non-linear layers are implemented by element-wise functions rather than matrix multiplication operations. 
As a result, the methods used for obfuscating linear layers are no longer applicable. This necessitates the development of specialized techniques to obfuscate non-linear layers effectively, ensuring that the entire model remains secure against potential attacks.

\subsection{Dense Layer}

The dense layers of neural networks perform straightforward matrix multiplication operations, making them suitable for masking techniques. 

For the $i$-th dense layer, let the weight matrix be $W_i$ and the bias be $b_i$. Given the input $X$, the computation performed by this layer is: $Y = X \cdot W_i^T + b_i.$

During the preprocessing phase, the TEE randomly samples transformation matrices $P$, $Q_{i-1}$, and $Q_i$ from the uniform distribution $U(-1, 1)$. These matrices are used to compute the masked versions of the input, weight, and bias as follows:
\begin{align}
    \tilde{X} = P X Q_{i-1}, \quad \tilde{W}_i = Q_{i-1}^{-1} W_i^T Q_i, \quad \tilde{b}_i = P \cdot (\mathbf{1}^T \cdot b_i) \cdot Q_i.
\end{align}

These masked results can be securely stored in untrusted environments.

During the inference phase, the dense layer operation on masked parameters and inputs is similar to the operation on the original version, which can be represented as  

\begin{equation}
\begin{split}
\tilde{Y} &= \tilde{X}\tilde{W}_i^T + \tilde{b}_i\\
&= (PXQ_{i-1})\cdot (Q^{-1}_{i-1}W^TQ_i) + P\cdot (\mathbf{1}^T \cdot b_i) \cdot Q_i\\
&= P(XW_i^T + b_i)Q_i.
\end{split}
\end{equation}

The output $P(XW_i^T + b_i)Q_i$ is equivalent to the original result but projected by the transformation matrices $P$ and $Q_i$. Since the two matrices are unknown to the adversary, the private values of $X$, $W_i$, and $b_i$ remain protected throughout the computation.

\subsection{Convolution Layer}

 The convolution layer can be regarded as a variant of matrix multiplication. Consider the $i$-th convolutional layer with operation $\mathbf{Y} = \text{Conv}(\mathbf{X}, \mathbf{W}_i)$, where $ \mathbf{X}$ represents the input tensor and $\mathbf{W}_i$ denotes the filter weight tensor. We can represent the filter $ \mathbf{W}_i $ as a structured hypermatrix:
$$
\mathbf{W}_i = \begin{bmatrix}\mathbf{w}_{11} & \mathbf{w}_{12} & \cdots &\mathbf{w}_{1d_2}
    \\ \mathbf{w}_{21} & \mathbf{w}_{22}& \cdots & \mathbf{w}_{2d_2}
    \\ \vdots & \vdots & \vdots & \vdots
    \\ \mathbf{w}_{d_31} & \mathbf{w}_{d_32} & \cdots & \mathbf{w}_{d_3d_2}
    \end{bmatrix},
$$
where each element $\mathbf{w}_{kl}$ represents the convolution kernel connecting the $l$-th input channel to the $k$-th output channel in $i$-th layer. $d_2$ and $d_3$ denote the number of input and output channels, respectively.

During the preprocessing phase, the masking process is similar to that of the dense layer. Given random invertible matrices $P, Q_{i-1}, Q_i$, we first apply hypermatrix multiplication to transform $\mathbf{W}_i$ into its masked version $ \tilde{\mathbf{W}}_i $ as follows:
\begin{equation}
\begin{split}
&\tilde{\mathbf{W}_i} = Q_{i-1}^{-1}\cdot \mathbf{W}_i \cdot Q_i\\
&\tilde{\mathbf{w}}_{sr} = \sum^{d_2}_{l=1}\sum^{d_3}_{k=1}(Q_{i-1}^{-1})_{sk}\cdot \mathbf{w}_{kl}\cdot (Q_i)_{lr}.
\end{split}
\end{equation}
The masked convolution layer is subsequently performed as:
\begin{equation}
\begin{split}
\text{Conv}(\tilde{\mathbf{X}}, \tilde{\mathbf{W}}_i) &= \text{Conv}(P\cdot \mathbf{X}\cdot Q_{i-1},\ Q^{-1}_{i-1}\cdot \mathbf{W}_i\cdot Q_i) \\
&= P\cdot \text{Conv}(\mathbf{X}, \mathbf{W}_i)\cdot Q_i.
\end{split}
\end{equation}
The bias vector $b_i = [b_1, \cdots, b_{d_3}]$ undergoes a similar transformation using the same masks: $ \tilde{\mathbf{b}}_i = P \cdot (\mathbf{1}^T \cdot \mathbf{b}_i) \cdot Q_i$. Analogous to the dense layer, the final output $ P \cdot \text{Conv}(\mathbf{X}, \mathbf{W}_i) \cdot Q_i $ of the inference phase represents a masked version of the true output.

\subsection{BatchNorm Layer}

Batch Normalization (BatchNorm) is a popular gadget that can accelerate convergence and simplify the backward propagation process. 
For an input vector $ X $, BatchNorm is computed as:
\begin{equation}
\begin{split}
\text{BatchNorm}(X) &= \frac{X - \mu}{\sigma} \cdot\gamma + \beta
\end{split}
\end{equation}
where $ \mu = \frac{1}{d} \sum_{i \in [d]} x_i $ is the mean and $ \sigma = \sqrt{\frac{1}{d} \sum_{i \in [d]} (x_i - \mu)^2 + \epsilon} $ is the standard deviation. This step can be rewritten as a linear transformation: $\text{BatchNorm}(X) = X \cdot W_{Bn} + b_{Bn}$, where
\begin{equation}
\begin{split}
&W_{Bn} = \begin{bmatrix}
   \frac{\gamma_1}{\sigma_1} & & \\
   & \ddots & \\
   & & \frac{\gamma_n}{\sigma_n}
   \end{bmatrix},\ 
   b_{Bn} = \left[\beta_1 - \frac{\mu_1 \gamma_1}{\sigma_1}, \cdots \beta_n - \frac{\mu_n \gamma_n}{\sigma_n}\right].
\end{split}
\end{equation}
During the inference phase, as the mean $ \mu $ and variance $ \sigma $ of BatchNorm are frozen, this layer can be conceptualized as a linear transformation that scales the input $ X $ using a scaling factor $ \gamma $ and shifts it by a bias term $ \beta $.

Given that BatchNorm is typically applied subsequent to convolution operations, computational efficiency can be optimized by combining the convolution kernel and BatchNorm parameters into a modified kernel $ \hat{W} = W_{Conv} \cdot W_{Bn} $ and bias $ \hat{b}=  b_{Bn} + b_{Conv} \cdot W_{Bn} $, thereby integrating the two operations into a single computational step. 
Notably, since these transformations are independent of input data, they can be completed during the preprocessing phase without compromising the efficiency of the inference phase.

\subsection{AvgPool Layer}

The average pooling (AvgPool) layer is a widely employed pooling mechanism for reducing the spatial dimensions of feature maps. It computes the average value within localized regions of the input feature map, generating output feature maps of reduced dimensionality. Since Avgpool operates on the input tensor in a manner that preserves linearity with respect to scalar multiplication and addition, it can be regarded as a linear transformation within the framework of tensor operations. This linear property allows Avgpool inherently compatible with our proposed masking pipeline. Formally, given the input tensor $\mathbf{X}$ is masked with random matrices $P$ and $Q$, we have:
\begin{align} \text{Avgpool}(P\cdot \mathbf{X} \cdot Q) = P \cdot \text{Avgpool}(\mathbf{X})\cdot Q. \end{align}
This formulation demonstrates that the transformations of AvgPool do not affect the integrity of the masks, thus allowing direct compatibility with other layers in the network.

\subsection{Flatten Layer}

The flatten layer is used to transform high-dimensional tensors (e.g., 4D tensors from convolutional layers) into 2D matrices. In  \MyScheme{}, to mask its output, we generalize the traditional flatten operation to 4D tensors, utilizing Kronecker products to preserve the relationships between the dimensions of the tensors.

We define the flatten operation $\text{Flatten}(\mathbf{X}, n)$ to transform the tensor $\mathbf{X}$ into a matrix starting from its $n$-th dimension. For a 4D tensor $\mathbf{X}$, the operation $\text{Flatten}(\mathbf{X}, 1)$ flattens the last three dimensions into a single dimension. A key property of two-dimensional flattening operation is:
\begin{align}
    \text{Flatten}(PXQ) = \text{Flatten}(X)(P^T \otimes Q),
\end{align}
where $P$, $X$, $Q$ are matrices and $\otimes$ denotes the Kronecker product.

When the input tensor $\mathbf{X}$ is masked using matrices $P$ and $Q$, resulting in $P \cdot \mathbf{X} \cdot Q$, the left mask $P$ operates on the first dimension and thus does not interfere with the flatten operation, allowing it to be extracted directly. The right mask $Q$ applies to the last three dimensions of $\mathbf{X}$. By employing the Kronecker product, the effect of $Q$ on the flattened tensor is captured by $I_n \otimes Q$.

In summary, for 4D tensors, the flatten operation with masking is generalized as:
\begin{equation}
\begin{split}
   \text{Flatten}(P \cdot \mathbf{X} \cdot Q, 1) = P \cdot \text{Flatten}(\mathbf{X}, 1) \cdot (I_n \otimes Q),
\end{split}
\end{equation}
where $I_n$ is the identity matrix of size $n$, with $n$ being the product of the last two dimensions of $\mathbf{X}$. The output of the flatten operation maintains compatibility with the masking framework.

\subsection{Activation Layer}

The activation layer transforms inputs through element-wise non-linear functions, enabling neural networks to learn complex patterns. However, if we continue using random matrix masking as in linear layers, applying activation functions to the obfuscated results would disrupt the reversibility of linear obfuscation, making it impossible to recover the original results. 
To address this challenge, \MyScheme{} first transforms the random matrix mask into a permutation matrix. This transformation is achieved by multiplying the obfuscated input with specially constructed matrices. To prevent potential leakage of the original input values during the transfer process from the random matrix mask to the permutation matrix mask, we perform an additional Kronecker product operation between the input and the random matrix. Since the Kronecker product possesses the mixed-product property, which combines the original matrix product and the Kronecker product in multiple Kronecker product operations, we can obtain the original result without exposing the original data. 
Here we present two cases about the ReLU function and the GELU function to demonstrate how to obfuscate the non-linear activation layers. 

\noindent
\textbf{ReLU Function.}
The ReLU function serves to determine the non-negative component of the input, defined as $\text{ReLU}(x) = \text{max}(x,0)$. 

During the preprocessing phase, the TEE first generates two masks:
\begin{align}
M_1 = \pi_3(\pi_1P^{-1} \otimes R_1), \quad M_2 = (Q^{-1}\pi_2 \otimes R_3)\,\pi_4,
\end{align}
where $R_1, R_2, R_3$ are drawn from the uniform distribution $U(0, 1)$, and $\pi_i$ are random permutation matrices which are square binary matrices with a single 1 in each row and column, and all other entries 0. These two masks are used to eliminate the original random mask while converting it into a random permutation mask. The notation $\otimes$ represents the Kronecker product operation. And we define $\bar{R} = R_1R_2R_3$ for convenience in subsequent discussions.

During the inference phase, \MyScheme{} first extends the dimension of the input $\tilde{X}$ using $R_2$, and then due to the mixed-product property of the Kronecker product, the masks are transformed as follows:
\begin{equation}
\begin{split}
Z &= M_1 \cdot (\tilde{X} \otimes R_2) \cdot M_2 \\
&= \pi_3(\pi_1P^{-1} \otimes R_1)\cdot (PXQ \otimes R_2)\cdot (Q^{-1}\pi_2 \otimes R_3)\,\pi_4 \\
&= \pi_3(\pi_1X\pi_2 \otimes \bar{R})\,\pi_4.
\end{split}
\end{equation}

Subsequently, the ReLU operation can be directly applied to $Z$. Because random permutations only shuffle the order without changing the values, and the ReLU exhibits positive homogeneity. Thus, we have $\text{ReLU}(Z) = \pi_3(\pi_1\text{ReLU}(X)\,\pi_2\otimes \bar R)\,\pi_4$.
Finally, $M_1^{-1}$ and $M_2^{-1}$ along with $R_2$ are employed to convert the random permutation matrix mask back to the previous random masks $P$ and $Q$ for computation in the subsequent layer. The aforementioned process is summarized in Alg.~\ref{alg:relu}. Compared to the original ReLU function which only performs element-wise comparison operations, four additional matrix multiplications are required, along with one Kronecker product and its inverse operation in \MyScheme{}. Fortunately, all of the aforementioned operations can be accelerated using GPU.

\begin{algorithm}
    \LinesNumbered
    \caption{$\prod_{\text{ReLU}}(\tilde{X})$}
    \label{alg:relu}
    \begin{flushleft}
    \textbf{Preprocessing:} TEE generates random matrices $\{R_i\}_{i=1}^3$ and random permutation matrices$\{\pi_i\}_{i=1}^4$, then calculates $M_1 = \pi_3(\pi_1P^{-1}\otimes R_1$), $M_2 = (Q^{-1}\pi_2\otimes R_3)\pi_4$; \\ 
    \textbf{Input:}  Tensor $\tilde{X} = PXQ$; \\
    \textbf{Output:}  Tensor $P\cdot \text{ReLU}(X) \cdot Q$;
    \end{flushleft}
    \begin{algorithmic}[1]
        \STATE Enlarge the matrix $\tilde{X} \otimes R_2$;
        \STATE Take the linear transformation to remove the original \\
        \quad Mask $Z = M_1\cdot (\tilde{X} \otimes R_2) \cdot M_2$;
        \STATE Compute ReLU($\cdot$) function $\text{ReLU}(Z) = \pi_3(\pi_1\text{ReLU}(X)\,\pi_2\otimes \bar R)\,\pi_4$;
        \STATE Reconstruct masks $M_1^{-1}\text{ReLU}(Z)M_2^{-1} = (P\cdot \text{ReLU}(X)\cdot Q) \otimes R_2$;
        \STATE Recover the result as $P\cdot \text{ReLU}(X)\cdot Q$.
    \end{algorithmic}
\end{algorithm}

\noindent
\textbf{GELU Function.}
The GELU function is widely regarded as an advanced alternative to the ReLU due to its smooth transition for values near zero. Although the implementation for GELU is more complex than ReLU, our approach adheres to the idea of obfuscating activation functions. Mathematically defined as $\text{GELU}(x) = x\cdot \Phi(x)$, where $\Phi$ represents the standard Gaussian cumulative distribution function, this probabilistic non-linear structure makes GELU a non-homogeneous function.  H, the random matrix $\bar R$ cannot be separated from the results after non-linear computation. To address this, we specially design $\bar R$ and employ index-selected matrices to extract valid activation values from the Kronecker product space.

During the preprocessing phase, TEE constructs a random matrix $R_2$ to mask the original input $X$, two masks to transform the original random mask of the input into random permutation matrices $M_1, M_2$.
Note that matrix $\bar R$ in the GELU function is constructed such that it contains exactly one element set to 1, ensuring the correctness of subsequent operations. Additionally, two further masks are established for reconstructing the final result:
\begin{equation}
M_3 = P\pi_1^TE_1\pi_3^T, \quad M_4 = \pi_4^TE_2\pi_2^TQ,
\end{equation}
$E_1$ and $E_2$ are index-selected matrices picking the corresponding elements from Kroneckor product. They satisfy the following relationship:
\begin{equation}
E_1\cdot \text{GELU}(\pi_1X\pi_2 \otimes \bar R)\cdot E_2 = \text{GELU}(\pi_1X\pi_2).
\end{equation}

During the inference phase, given the input $\tilde{X}$, similar to the process in ReLU, \MyScheme{} first applies an expansion, followed by a mask transformation, and then computes GELU function obtained $\text{GELU}(Z)$. Subsequently, unlike the ReLU procedure, we utilize $M_3$ and $M_4$ to recover the result obtained through the random matrices $P$ and $Q$ as follows:
\begin{equation}
    P\cdot \text{GELU}(X)\cdot Q = M_3\cdot \text{GELU}(Z)\cdot M_4.
\end{equation}

Throughout this process, our scheme introduces additional computational overhead that includes four matrix multiplications and one Kronecker product operation.

\subsection{Attention Block}
The self-attention layer, which enables a model to dynamically weigh and capture relationships between all positions in a sequence, is the core of the Transformer architecture~\cite{vaswani2017attention}. It operates through Query ($Q$), Key ($K$), and Value ($V$) projections combined with scaled dot-product attention. For an input tensor $X$, the self-attention operation can be formalized as follows:
\begin{align}
   &Q = XW_q, \quad K = XW_k, \quad V = XW_v, \\
   &\text{Attention}(Q, K, V) = \text{softmax}\left(\frac{QK^T}{\sqrt{d}}\right) V,
\end{align}
where $W_q, W_k, W_v$ is the self-attention weights. In large-scale language models, a commonly employed mechanism is the Multi-Head Self-Attention (MHA) block, which represents an advanced and sophisticated extension of the standard attention layer. MHA can be formalized as:
\begin{align}
   &\text{MHA}(X) = (\text{head}_0||\cdots||\text{head}_{h-1})W_o, \\
   &\text{head}_i = \text{Attention}(Q_i, K_i, V_i) \text{ for } i \in [h].
\end{align}

During the preprocessing phase of \MyScheme{}, the weight matrices of the attention block are obfuscated as:
\begin{equation}
\begin{split}
\tilde{W}_{q_i} &= N^{-1}W_{q_i}P_i^T, \quad \tilde{W}_{k_i} = N^{-1}W_{k_i}P_i^{-1}, \\
\tilde{W}_{v_i} &= N^{-1}W_{v_i}S_i, \quad \tilde{W}_o = \mathcal{S}^{-1}W_{o}N,
\end{split}
\end{equation}
where $P_i, S_i$ represent random masks for each attention head $\text{head}_i$, $\mathcal{S}$ denotes a block diagonal matrix composed of submatrices $S_i$, and $N$ is a random mask correlated with layer's input. These transformations can be executed during the preprocessing phase, necessitating four matrix multiplications.

During the inference phase, we obfuscate the input tensor $X$ through multiplication with a random permutation matrix $\pi$ and a random matrix $N$, represented mathematically as $\tilde{X} = \pi X N$. This formulation is strategically designed to enhance computational efficiency while facilitating softmax operations in transformer architectures. As demonstrated by Xu et al.~\cite{xu2024permutation}, permutation matrices possess the valuable property of guaranteeing that the softmax operation merely rearranges rows and columns without cross-mixing them. This mathematical property is critical for preserving the functional integrity of attention mechanisms while maintaining model privacy. Then, the transformed query, key, and value matrices are calculated as
\begin{equation}
\begin{split}
\tilde{Q}_i &= \tilde{X} \tilde{W}_{q_i} = \pi X W_{q_i} P_i^T  = \pi Q_i P_i^T, \\
\tilde{K}_i &= \tilde{X} \tilde{W}_{k_i} = \pi X W_{k_i} P_i^{-1} = \pi K_i P_i^{-1}, \\
\tilde{V}_i &= \tilde{X} \tilde{W}_{v_i} = \pi X W_{v_i} S_i = \pi V_i S_i.\\
\end{split}
\end{equation}
The attention score for each head in the multi-head attention block can be calculated as:
\begin{equation}
\begin{split}
\text{head}_i' &= \text{softmax}(\frac{\tilde{Q}_i\tilde{K}_i^T}{\sqrt{d}})\tilde{V}_i \\
&= \text{softmax}\left(\frac{\pi Q_iK_i^T \pi^T}{\sqrt{d}}\right) \pi V_iS_i \\
&= \pi\, \text{head}_i\cdot S_i.
\end{split}
\end{equation}

Due to the distinctive properties of permutation matrices when applied to query-key interactions within the softmax function, these mask matrices can be effectively separated from the function. This mathematical property enables the scores of all attention heads to be integrated within the MHA block according to the following formulation:
\begin{equation}
\begin{split}
\text{MHA}(\tilde{X}) &= (\text{head}'_0||\cdots||\text{head}'_{h-1})\cdot \tilde{W}_o \\
&= \pi(\text{head}_0||\cdots||\text{head}_{h-1})\mathcal{S}\cdot \tilde{W}_o \\
&= \pi\, \text{MHA}(X)N.
\end{split}
\end{equation}

The skip connection architecture inherent in transformer models requires that the mask $N$ for the output weight matrix $W_o$ be identical to the mask applied to the input tensor $X$. Finally, the step of the attention block can thus be succinctly expressed as $F(P X N) = P \cdot F(X) \cdot N$, where $P$ and $N$ represent the masking matrices and $F(\cdot)$ represents the normal MHA operation. This formulation demonstrates a critical advantage of our approach that the obfuscation operations impose no additional computational overhead during the inference phase.

\subsection{Layer Normalization}
Layer Normalization (LayerNorm) constitutes an essential component of the transformer architectures. The LayerNorm operation can be formally defined as:
\begin{align}
\text{LayerNorm}(X) = \gamma \, \text{Norm}(X) + \beta,
\end{align}
where $\gamma$ and $\beta$ represent learnable scaling and shifting parameters, respectively. The main difference between LayerNorm and BatchNorm lies in the requirement for an additional normalization step for each layer, rather than directly utilizing frozen means and variances to complete the computation. Since normalization is not a linear operation, this characteristic causes the masking approach employed in BatchNorm incompatible with LayerNorm. To address this limitation, our framework incorporates an additional obfuscation step for the input of LayerNorm.

During the preprocessing phase, we first introduce a gadget matrix $G = \lambda I + \mathbf{r}^T \cdot \mathbf{1}$, where $I$ is an identity matrix, $\mathbf{r}$ is a random vector, and $\lambda$ is a random scalar. This design ensures that $\text{Norm}(X \cdot G) = \text{Norm}(X)$. 
Next, we obfuscate the gadget matrix $G$ to obtain $\tilde{G} = N^{-1} G\pi_2$.

During the inference phase, upon receiving the LayerNorm input $\tilde{X}$, i.e., the output from the prior self-attention layer, we first execute the mask replacement operation as follows:
\begin{equation}
\tilde{X} \cdot \tilde{G} = (\pi_1XN) \cdot (N^{-1}G\pi_2) = \pi_1XG\pi_2,
\end{equation}
$\tilde{G}$ represents the masked gadget matrix. Then, we can derive the obfuscated LayerNorm operation as:
\begin{equation}
\gamma\, \text{Norm}(\tilde{X} \tilde{G}) + \beta = \gamma (\pi_1\cdot \text{Norm}(X)\cdot \pi_2) + \beta,
\end{equation}
where $\pi_i$ denotes batch permutation matrices, with the batch size corresponding to the number of input samples. The processes for handling the scaling parameter $\gamma$ and shifting parameter $\beta$ are identical to those in BatchNorm.

\begin{algorithm}[t]
\caption{$\prod_{\text{Amulet}}(X, W)$}
\label{alg:allprocess}
\KwIn{Input $X$, model weights $\{W_i\}_{i=1}^{n}$, where $n$ is the number of linear layers}
\KwOut{Inference result $Y$}

\tcc{$\triangledown$ \textbf{One-time Model Obfuscation (within TEE)}}
TEE generates random invertible matrices $\{Q_i\}_{i=0}^n$\;
TEE obfuscates weight matrices: $\tilde{W_i} = Q_{i-1}^{-1} W_i Q_i$, $\forall i \in [1,n]$

\tcc{$\triangledown$ \textbf{Input-independent Preprocessing (within TEE)}}
TEE generates two random matrices $P, T$\;
TEE computes $T_Y = P \cdot \mathbf{L_1}(T, W_1) \cdot Q_1$ for input layer\;\tcp{Input Padding Generation}
TEE generates random $\{\pi_j\}_{j=1}^4$, $\{R_j\}_{j=1}^3$ for each non-linear layer and initializes masks \tcp*[r]{Non-linear Obfuscation}

\tcc{$\triangledown$ \textbf{Inference (cross TEE \& GPU)}}
User transfers input $X$ to TEE from untrusted memory\;
TEE obfuscates input: $\tilde{X} = P(X-T)Q_1$ and sends it to untrusted memory \tcp*[r]{Input Masking}
User computes $\tilde{Y}_1 = \mathbf{L_1}(\tilde{X}, \tilde{W_1}) + T_Y$\;
User performs forward pass through all layers\;
User obtains $\tilde{Y}_n = P Y_n Q_n$ and sends to TEE\;
TEE recovers $Y_n$, computes $\arg\max(Y_n)$ and returns $Y$ to untrusted memory \tcp*[r]{Result Restoration}
\end{algorithm}

\subsection{Obfuscating the Input and Output}

After obfuscating all intermediate layers of the model using the aforementioned approach, a potential vulnerability remains regarding the input layer. Since the adversary knows its input and can choose arbitrary values, it might potentially conduct chosen-plaintext attacks by constructing specific inputs to analyze model information. To address this issue, our scheme implements an additional one-time pad for model inputs, preventing users from influencing the obfuscated results through input selection. The notation $\textbf{L}_i$ represents the $i$-th linear layer in the neural network.

For improved clarity, we further partition the input-independent preprocessing phase into two distinct parts: weights obfuscation and one-time-pads (OTPs) generation. For the weights obfuscation part, the TEE generates random invertible matrices $Q_i$ and obfuscates the weight matrix of each linear layer $W_i$ into $Q_{i-1}^{-1}W_iQ_i$. For the OTPs generation part, the TEE generates fresh random matrices $P$ and $T$, computes $T_Y = P\cdot \mathbf{L_1}(T,W_1)\cdot Q_1$ to mask the inputs received during the inference phase, and prepares additional random matrices $\{\pi_j\}_{j=1}^4$ and $\{R_j\}_{j=1}^3$ for each non-linear layer. Note that the OTPs are generated fresh for each inference, but since the generation process is independent of the input, they can be generated during the preprocessing phase. The masks $P$, $Q_0$, $Q_n$ are always stored in the TEE for use in the subsequent inference phase.

During the inference phase, the user first sends the original input $X$ to the TEE. The TEE masks $X$ to obtain $\tilde{X} = P(X-T)Q_0$ and returns it to untrusted memory. The user then feeds $\tilde{X}$ to the obfuscated first layer, resulting in $\tilde{Y}_1 = \mathbf{L_1}(\tilde{X}, \tilde{W}) + T_Y$. Starting from $\tilde{Y}_1$, the computation proceeds through each layer as described earlier in this chapter, until reaching the model's final layer $\mathbf{L}_n$. After obtaining $\tilde{Y}_n = P\cdot \mathbf{L_n}(X, W)\cdot Q_n$, the user returns this result to the TEE, which removes the obfuscation using $P$ and $Q_n$, computes the remaining steps (such as $\textsf{argmax}$ for classification tasks), and returns the final result to the user. Alg.~\ref{alg:allprocess} summarizes the complete steps required for performing a single inference using \MyScheme{}.

\section{Security Analysis}

In this section, we aim to establish the security of \MyScheme{} from an information-theoretic standpoint.
The analysis is divided into three parts: (1) input obfuscation, (2) intermediate linear layer obfuscation, (3) intermediate non-linear layer obfuscation.

\begin{theorem} \label{theorem:1}
(Input Obfuscation) Let $W$ denote the weight matrix of the first layer. Assume that $P$, $Q$, and $S$ are independently and uniformly sampled invertible matrices, and that the random matrices $\{T_i\}_{i=1}^t$ are independently and uniformly selected for each round. Then, for any polynomial number of rounds $t$, it holds that:
\[
I(W; \mathcal{O}_1, \dots, \mathcal{O}_t) = 0,
\]
where $\mathcal{O}_i = \{X_i, \tilde{X}_i, \tilde{T}_i, \tilde{W}\}$ with $\tilde{X}_i = P(X_i - T_i)Q,\, \tilde{T}_i = PT_iWS$ and $\tilde{W} = Q^{-1}WS$.
\end{theorem}

\begin{theorem}\label{theorem:2}
(Intermediate Non-linear Obfuscation) Let $X$ denote the original input matrix of the non-linear layer. Assume that $P, Q$ are independently and uniformly sampled invertible matrices, and the random matrices $\{\pi_j\}_{j=1}^4$ and $\{R_j\}_{j=1}^3$ are independently and uniformly selected for each round. So, it holds that:
\[
I(X; \mathcal{O}) = 0,
\]
where $\mathcal{O} = \{\tilde{X}, M_1, M_2, M_3, M_4, R_2\}$.
\end{theorem}

\begin{theorem}\label{theorem:3}
(Intermediate Linear Obfuscation) Let $X$ denote the original input matrix of the non-linear layer, $W$ denote the weight matrix of the intermediate linear layer. Assume that $P$, $Q$, and $S$ are independently and uniformly sampled invertible matrices. Then, for any polynomial number of rounds $t$, it holds that:
\[
I(W; \mathcal{O}_1, \dots, \mathcal{O}_t) = 0,
\]
where $\mathcal{O}_i = \{\tilde{X}_i, \tilde{W}\}$ with $\tilde{X}_i = PY_iQ$ and $\tilde{W} = Q^{-1}WS$.
\end{theorem}

\begin{proof}[Proof (Sketch).]

For the input layer and the $i$-th linear layer, the adversary can observe primarily two types of information: the obfuscated weight matrix $\tilde{W}$ and the obfuscated input $\tilde{X}$ to that layer. For $\tilde{W}$, since the original $W$ could be any matrix in the entire space with equal probability due to the randomness of $Q_i$ and $Q_{i+1}$, we maintain the entropy $H(W\mid \tilde{W}) = H(W)$. For $\tilde{X}$, it is obfuscated by a fresh and uniformly random one-time mask during each inference (where the mask consists of the padding matrix $T$ for the input layer, and $M_1$ and $M_2$ from the previous non-linear layer for others). The mask renders $\tilde{X}$  statistically independent of $\tilde{W}$. Since the adversary does not know the mask, this observation provides no additional information about $\tilde{W}$. 
For non-linear layers without weight matrices, a similar analysis demonstrates that the adversary cannot extract information about the true input (i.e., the output from the previous linear layer) from the obfuscated input $\tilde{X}$ of that layer. Therefore, for each query to the obfuscated model, no individual layer leaks information about $W$.
Formally, the mutual information between $W$ and the observations of the adversary is zero. Each layer in the model achieves perfect information-theoretic security.
\end{proof}

Due to space limitations, detailed proofs of the above theorems are provided in Appendix~\ref{sec.ana}. 

These theorems establish that an adversary with access to all observable variables gains zero information about the weight matrix $W$ of each individual layer. When multiple layers are combined, sequential execution without parallel operations ensures that no cross-layer information leakage occurs. Hence, composing individual layers to form the complete model will not compromise the overall security. This aligns with the standard sequential composition theorem~\cite{lindell2017tutorials}. Further, as the information gain about $W$ from a single query is zero, even after an arbitrary number of queries, the adversary cannot obtain additional information about $W$. As a result, the knowledge exposed from the obfuscated model is equivalent to that of a completely black-box, label-only model.

Note that the security guarantee holds even against a malicious adversary. The adversary may choose arbitrary inputs to perform chosen-plaintext attacks or known-plaintext attacks. However, due to the presence of one-time masks for the input of each layer, any input chosen by the adversary will be transformed into random values, making them indistinguishable from any other input. In Appendix~\ref{sec.ana}, we show that even in the extreme case where the input is an all-zero matrix, the adversary also cannot obtain information about $W$.

\begin{figure*}
    \centering
    \includegraphics[width=0.95\textwidth]{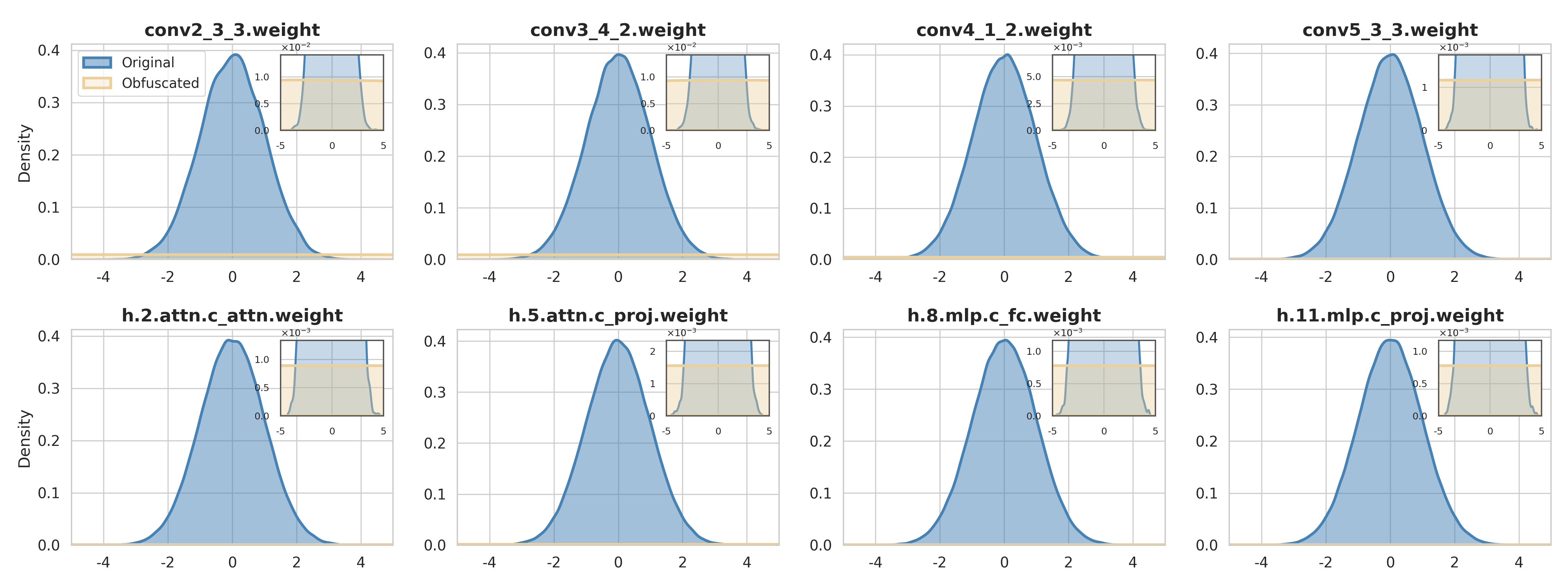} 
    \caption{Distribution of masked parameter and public parameter.}
    \label{fig:histogram}
\end{figure*}

\begin{figure}
    \centering
    \includegraphics[width=0.95\linewidth]{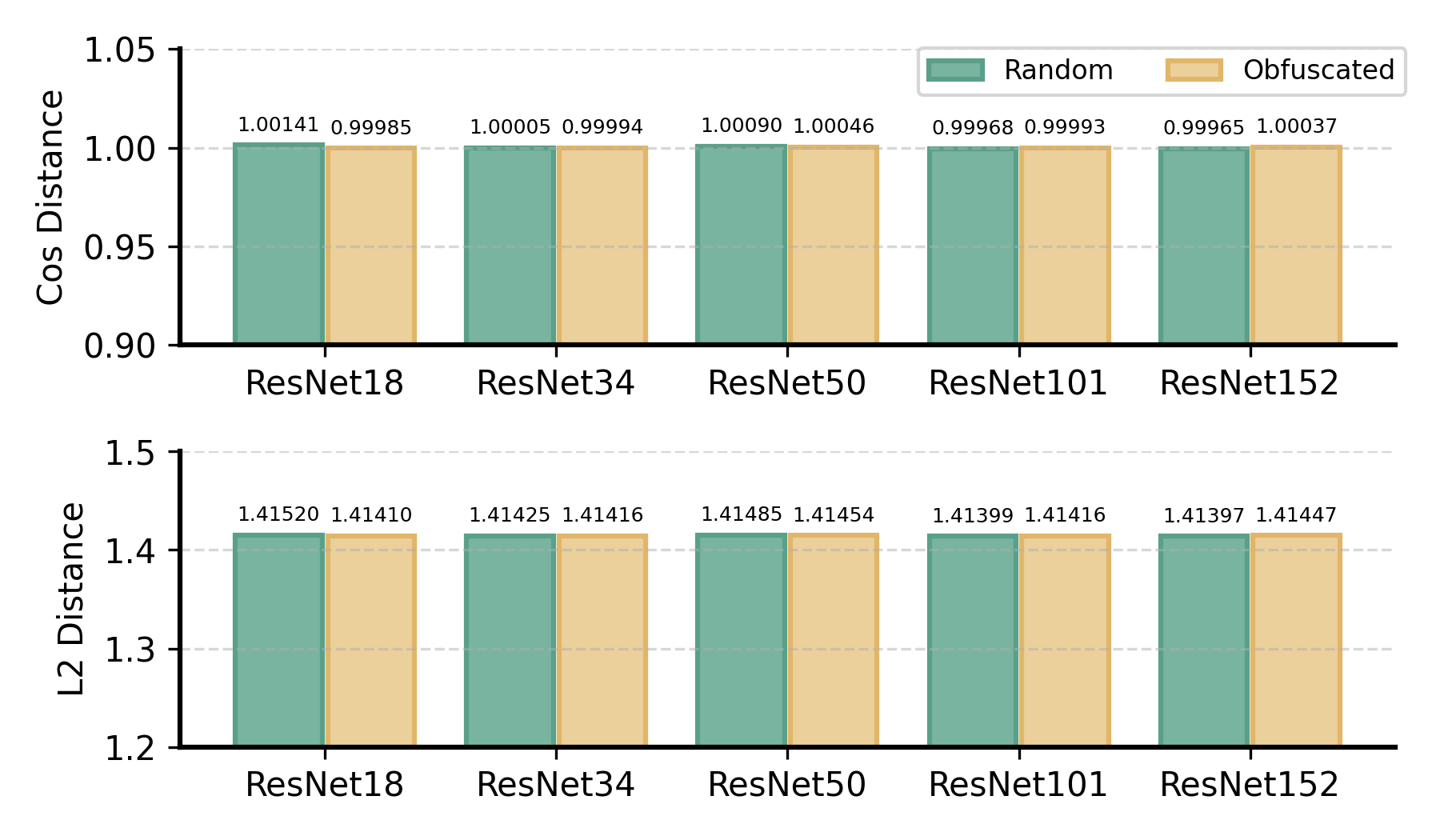} 
    \caption{The statistic comparison between Amulet obfuscated weights vs. original weights and random matrices vs. original weights.}
    \label{fig:statistic_ana}
\end{figure}

\noindent
\textbf{Potential Vulnerabilities.}
Recent works have explored potential vulnerabilities in prior obfuscation-based solutions. These works can be categorized into two types: statistical analysis of obfuscated parameters~\cite{zhang2024no, wang2025arrow}, and side-channel attacks targeting TEEs~\cite{yuan2024hypertheft}. For the former, the adversary may recover weights by leveraging distribution similarity between public and obfuscated private models~\cite{zhang2024no}, or by analyzing the directional similarity of weights before and after obfuscation~\cite{wang2025arrow}.
However, as demonstrated in our previous analysis, in \MyScheme{}, the information accessible to adversaries—namely, the obfuscated weights and intermediate results—has zero mutual information with the original weights. In other words, the adversary cannot leverage known information to infer details about either the original weights or the masks, rendering such attacks ineffective against our scheme.
In Fig.~\ref{fig:histogram} and ~\ref{fig:statistic_ana}, we analyze the parameter statistical properties of the model before and after obfuscation. We observe that the two distributions exhibit significantly different statistical characteristics. This result intuitively demonstrates that the obfuscated parameters have lost their correlation with the original parameters, thereby providing resistance against statistical analysis-based attacks.

Side-channel vulnerabilities remain a critical concern for TEE-based systems, with recent studies demonstrating the feasibility of attacking TEE-shielded DNN models~\cite{yuan2024hypertheft}. While existing countermeasures—such as memory randomization\cite{wichelmann2024obelix} and system obfuscation\cite{duy2025incognitos}—effectively reduce such leakage, \MyScheme{} is designed to be orthogonal to these techniques. By integrating \MyScheme{} with these mitigation strategies, we can establish a defense-in-depth architecture that offers comprehensive protection for DNNs within TEEs

Generally, by minimizing TEE usage and employing a provably secure obfuscation scheme, \MyScheme{} delivers enhanced security guarantees relative to other solutions, while remaining resilient against both types of attack methodologies.
\section{Evaluation}
We evaluate the performance of \MyScheme{} across four dimensions: inference efficiency, accuracy, preprocessing overhead, and storage overhead.

\begin{figure*}
    \centering
    \includegraphics[width=0.95\textwidth]{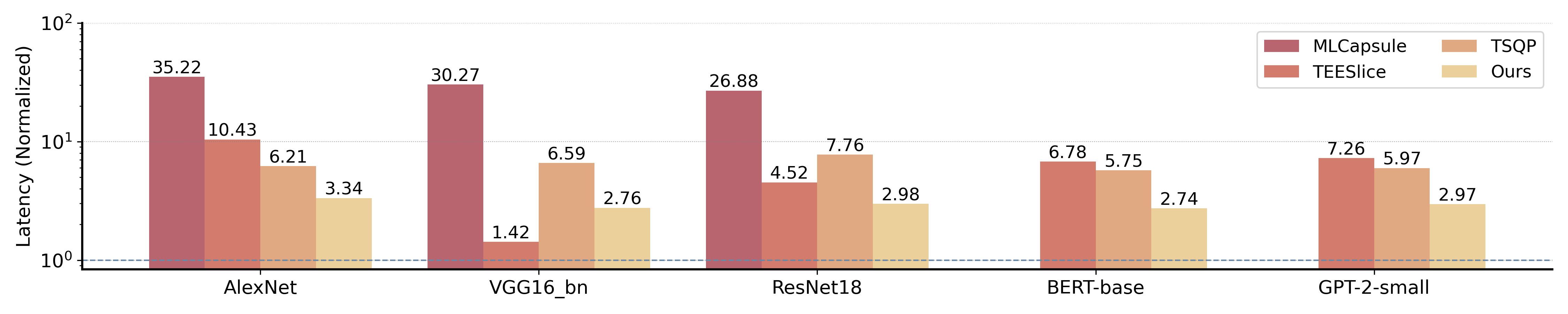} 
    \caption{Inference Latency across different models and methods. The results are normalized to the inference latency of the unprotected model (blue line). Since the Transformer model is too large to entirely load into the SGX enclave on our device, the corresponding results of MLCapsule are not included.}
    \label{fig:latency}
\end{figure*}
\begin{figure*}
    \centering
    \includegraphics[width=0.95\textwidth]{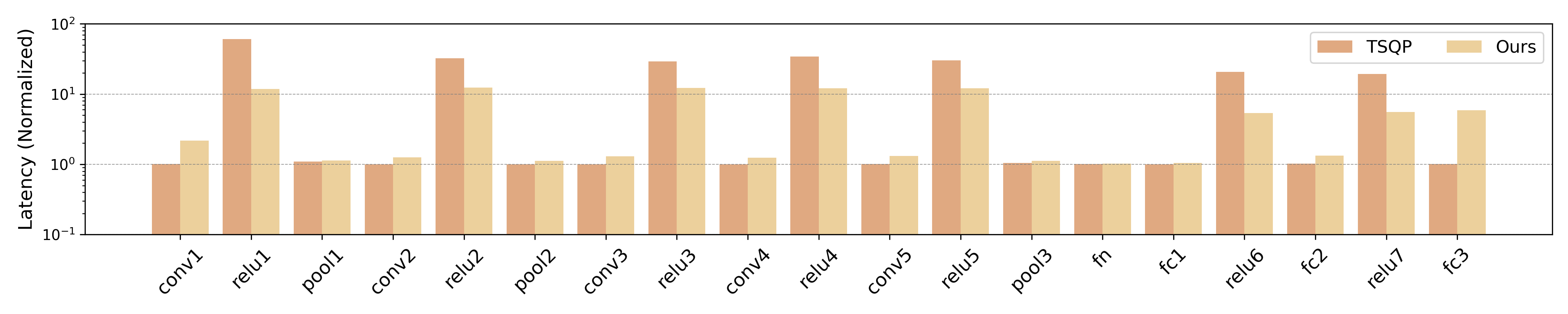} 
    \caption{Layers-to-layers inference latency of AlexNet. The results are normalized to the latency on the unprotected model. }
    \label{fig:latency2}
\end{figure*}
\begin{figure}
    \centering
    \includegraphics[width=0.98\linewidth]{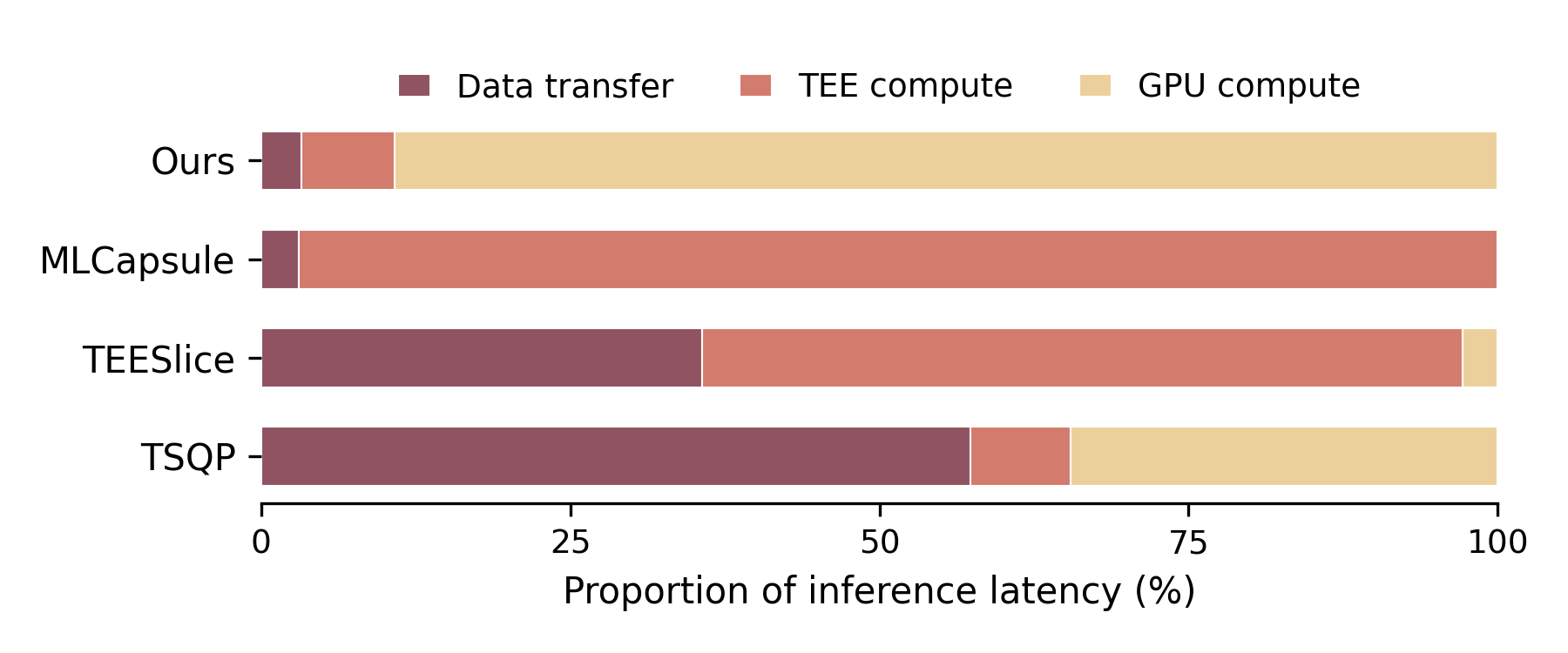} 
    \caption{Breakdown of inference latency across different methods.}
    \label{fig:breakdown}
\end{figure}

\subsection{Experimental Configuration}
In this paper, we use Intel SGX as an example of TEE implementation, primarily due to its accessibility and developer-friendly environment. Note that since our work focuses on utilizing TEEs as secure enclaves rather than delving into their internal implementation details, our approach can be applied to other TEEs without requiring new designs.

All experiments were conducted on a laptop equipped with an Intel Core i7-10870H CPU with SGX1 support, an NVIDIA GeForce RTX 2070s GPU with 8GB VRAM, and 16GB of system memory. The implementation consists of approximately 4,000 lines of code in PyTorch and C++, modified from TEESlice~\cite{zhang2024no}.

\noindent
\textbf{Models and Datasets.} We evaluate \MyScheme{} on 12 popular models, including AlexNet~\cite{krizhevsky2012imagenet}, VGG16 with BatchNorm~\cite{simonyan2015very}, ResNet-(18, 34, 50, 101, 152)~\cite{he2016deep}, BERT~\cite{devlin2019bert}, and GPT-2-(small, medium, large, xl)~\cite{radford2019language}. For AlexNet, VGG16, and ResNet-18, we use the CIFAR-10 and CIFAR-100 datasets~\cite{krizhevsky2009learning}. For BERT and GPT-2, we use the SST2 dataset~\cite{wangglue}. We compare the performance of our approach against other baseline methods across these models. Additionally, we conduct experiments on the ResNet series using ImageNet~\cite{russakovsky2015imagenetdataset} to assess how model size affects performance while maintaining similar architectural structures.

\subsection{Inference Latency}

\textbf{Overall Performance.} Inference latency is the most critical performance metric for evaluating our approach. We compared \MyScheme{} with TSQP\footnote{Since TSQP is optimized for QNNs, for fair comparison, we replace all operations in TSQP with standard floating-point operations to ensure it performs the same computational tasks as other methods.}, TEESlice, and MLCapsule. Additionally, we established performance boundaries by measuring two baselines: executing the entire inference within the TEE as a lower bound, and using GPU on unprotected models as an upper bound. By default, all experiments are conducted with an input batch size of 1. Note that for the MLCapsule approach, which utilizes TEE to complete the entire inference process, the Transformer model was excluded from the performance comparison results due to its substantial memory requirements exceeding the capacity of our device's enclave.

\MyScheme{} achieves an 8$\times$ to 9$\times$ speedup compared to the baseline method that completes the inference process entirely within the TEE. Under identical conditions, it is only 2.8$\times$ to 4.8$\times$ slower than unprotected model inference with GPU acceleration. Furthermore, across all tested model architectures, \MyScheme{} outperforms existing partition-based methods, delivering 1.5$\times$ to 3$\times$ faster latency compared to these approaches.

\begin{table}[ht!]
\caption{For the preprocessing phase, we mainly focus on the additional storage overhead introduced by model obfuscation. For the inference phase, we mainly focus on the additional computational overhead resulting from obfuscation.}
\label{tab:ops}
\centering
\resizebox{0.9\linewidth}{!}{
    \begin{tabular}{l|c|c}
        \toprule
            \textbf{Layers} & \makecell{\textbf{Extra storage} \\ \textbf{in Preprocessing}} & \makecell{\textbf{Extra operation} \\ \textbf{in Inference}} \\
        \midrule
        Conv/Dense & 1 Expanded Bias & None \\
        BatchNorm & 1 Expanded Bias & None \\
        Avgpool  & None & None \\
        Flatten  & None & None \\
        ReLU & 5 Matrices & 4 Mult + 2 Kron \\
        GELU & 5 Matrices & 4 Mult + 1 Kron \\
        MHA & 4 Expanded Biases & None \\
        LayerNorm & 1 Matrix & 1 Mult \\
        \bottomrule
    \end{tabular}
}
\end{table}

\noindent
\textbf{Detailed Analysis.}
In Table~\ref{tab:ops}, we summarize the additional overhead introduced by \MyScheme. For the preprocessing phase, since this overhead is relatively insensitive to time constraints, especially considering that the obfuscation of linear layer model parameters represents a one-time cost, we primarily focus on the additional storage requirements. For the inference phase, since \MyScheme does not require loading the model into the TEE, the extra memory usage is not significant; instead, we concentrate on computational overhead. As can be observed, linear layers incur minimal additional storage and computational costs. While non-linear layers introduce more substantial overhead, our approach still demonstrates performance advantages compared to baseline methods.

Fig.~\ref{fig:latency} shows the overall inference latency of \MyScheme{} and baseline methods across different models. The results are transformed as the ratio of inference latency between obfuscated and unprotected models. As can be seen, \MyScheme{} incurs only 2.8-4.8$\times$ overhead compared to direct plaintext inference on GPU. Compared to MLCapsule, which executes inference entirely within TEEs, \MyScheme{} achieves approximately 8-9$\times$ performance improvement. Moreover, compared to the recent TEESlice that protects only partial layers and the obfuscation-based TSQP, \MyScheme{} demonstrates approximately 2-3$\times$ improvement.
Fig.~\ref{fig:latency2} compares the layer-wise latency of \MyScheme{} and TSQP across all layers in AlexNet. For linear layers, both methods exhibit similar latency as they utilize GPU acceleration for these computations. However, for non-linear layers, \MyScheme{} demonstrates a significant performance advantage, as TSQP requires transferring non-linear layer inputs to the TEE for processing, resulting in time wasted on data transfer.

We further make a more in-depth analysis of the overhead of non-linear layers. Since most existing partition-based methods implement non-linear layers through trivial execution within TEEs, we used this approach as our baseline. Fig.\ref{fig:act} presents the running time of both methods for ReLU and GELU, the two most commonly used non-linear activation functions, across varying input sizes. Note that for the TEE-based method, the overhead includes both data transfer time between untrusted and trusted memory and the computation time of TEE, whereas our \MyScheme{} only incurs the time required for GPU-based computation of these operations. As demonstrated, our method achieves a 1.5-3$\times$ speedup compared to the basic approach. However, as the input size increases, our advantage gradually diminishes due to the additional matrix multiplication and Kronecker product operations required by our method.

Furthermore, to fully illustrate the benefits of \MyScheme, we conducted a more detailed analysis of the composition of inference time overhead. Fig.~\ref{fig:breakdown} presents the CNN model's proportional time spent by \MyScheme{} and several baseline methods on three components: data transfer between untrusted and trusted memory, TEE computation, and GPU computation. Notably, other methods spend considerable time on data transfer, with TEE handling numerous computational steps. In contrast, our method only requires minimal information transfer (input and output results) between untrusted and trusted memory and eliminates the need for TEE involvement during the inference phase, thereby more effectively leveraging GPU acceleration capabilities. These two factors are the primary drivers of our performance improvements.

\begin{figure}
    \centering
    \includegraphics[width=0.85\linewidth]{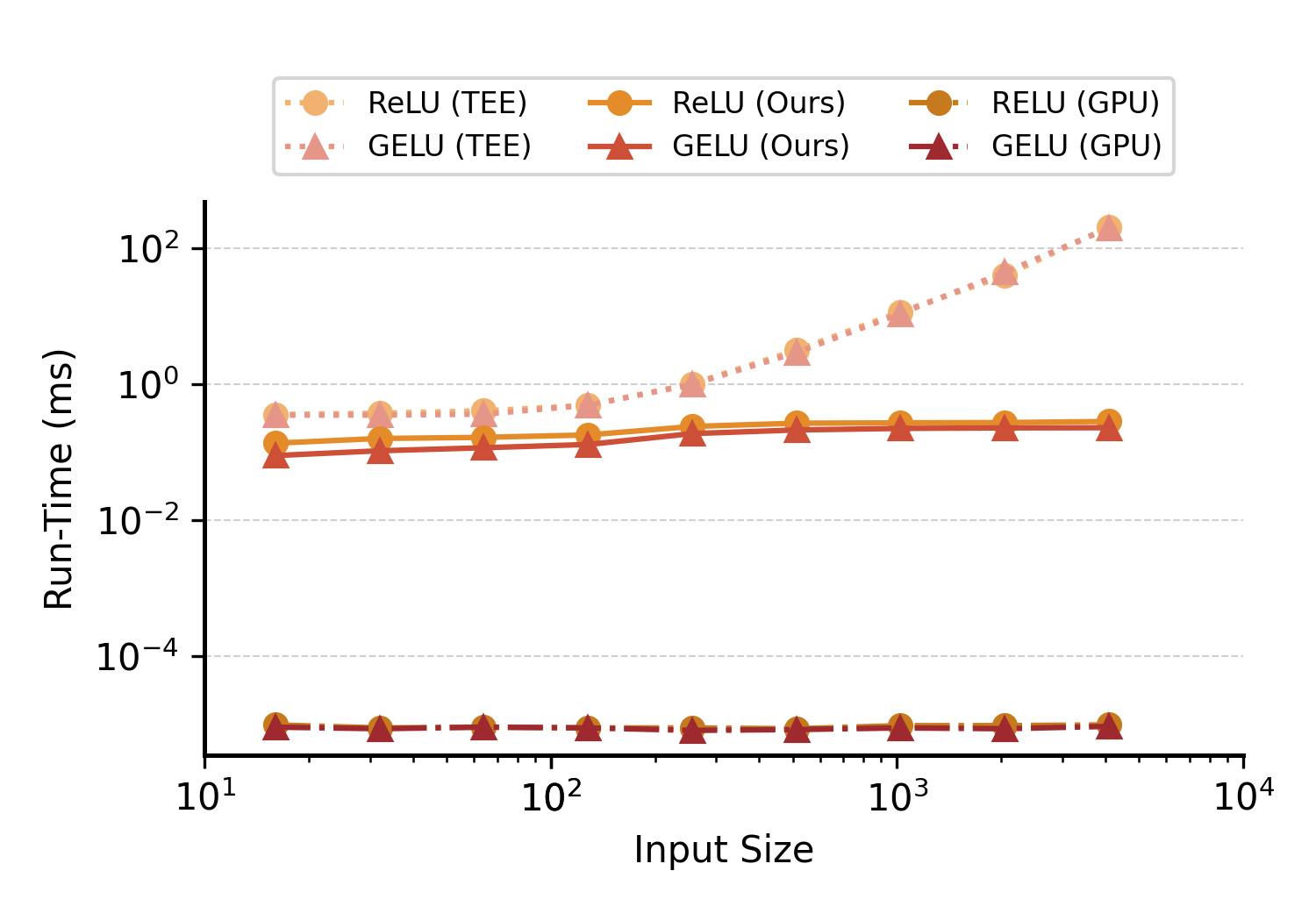} 
    \caption{The runtime of ReLU and GELU operations across varying input sizes.}
    \label{fig:act}
\end{figure}

\noindent
\textbf{Scalability.}
Here we analyze in detail the performance of \MyScheme across models of different scales. Notably, for models of varying scales, \MyScheme maintains consistent inference latency overhead of approximately 2.8-4.8$\times$ relative to plaintext inference, without significant increases due to model size. This result primarily stems from the constant number of interactions between untrusted and trusted memory in our scheme, and the predictable additional overhead for the obfuscated model. Specifically, as we summarized in Table~\ref{tab:ops}, for linear layers, \MyScheme introduces no additional operations compared to plaintext inference, thus no extra overhead. For non-linear layers, the results in Fig.~\ref{fig:act} show that as input dimensions increase, \MyScheme and plaintext inference maintain essentially the same growth rate in runtime. 

Fig.~\ref{fig:resgpt} presents the impact of model scale on inference latency. As observed, for GPT-2 family models ranging from the smallest GPT-2-small to the largest GPT-2-xl, the latency remains at approximately 2.9$\times$ that of plaintext inference, while ResNet models exhibit at most 4.8$\times$ overhead, demonstrating the scalability of our scheme.

\begin{table*}[ht]
    \centering
    \caption{Inference accuracy of different models with batch\_size=1.}
    \label{tab:accuracy} 
    \resizebox{0.95\textwidth}{!}{
        \begin{tabular}{l|cccccccccccc} 
            \toprule 
            \multirow{2}{*}{Models} & \multirow{2}{*}{AlexNet} & \multirow{2}{*}{\makecell[c]{VGG 16\\with bn}} & \multicolumn{5}{c}{ResNet} & \multirow{2}{*}{\makecell[c]{BERT\\base}} & \multicolumn{4}{c}{GPT-2} \\ \cmidrule(r){4-8} \cmidrule(r){10-13}
            & & & 18 & 34 & 50 & 101 & 152 & & -s & -m & -l & -xl \\
            \midrule 
            Difference & 1.9e-4 & 1.3e-4 & 1.4e-4 & 2.3e-4 & 3.4e-4 & 3.8e-4 & 3.6e-4 & 4.0e-4 & 2.7e-08 & 5.8e-08 & 8.9e-08 & 1.8e-7 \\
            Class Acc. & 100\% & 100\% & 100\% & 100\% & 100\% & 100\% & 100\% & 100\% & 100\% & 100\% & 100\% & 100\%  \\
            \bottomrule 
        \end{tabular}
    }
\end{table*}

\begin{table}[ht]
    \centering
    \caption{Accuracy of different model protection methods on three models (AlexNet, VGG16\_bn, ResNet-18) evaluated on the CIFAR-10 and CIFAR-100 datasets.}
    \captionsetup{skip=8pt}
    \label{tab:bigacc}
    
    \resizebox{0.95\linewidth}{!}{
        \begin{tabular}{l|cccc}
        \toprule
        Model & No-shield & TEESlice & TSQP & Amulet \\
        \midrule
        \multicolumn{5}{c}{CIFAR-10} \\
        \midrule
        AlexNet   & 82.18\% & 84.72\% & 81.78\% & 82.18\% \\
        VGG16\_bn & 90.14\% & 91.50\% & 89.98\% & 90.14\% \\
        ResNet18  & 95.27\% & 93.45\% & 94.75\% & 95.27\% \\
        \midrule
        \multicolumn{5}{c}{CIFAR-100} \\
        \midrule
        AlexNet   & 55.63\% & 60.11\% & 55.41\% & 55.63\% \\
        VGG16\_bn & 72.50\% & 72.54\% & 72.08\% & 72.50\% \\
        ResNet18  & 74.68\% & 71.53\% & 74.19\% & 74.68\% \\
        \bottomrule
        \end{tabular}
    }
\end{table}

\begin{figure}
    \centering
    \includegraphics[width=0.9\linewidth]{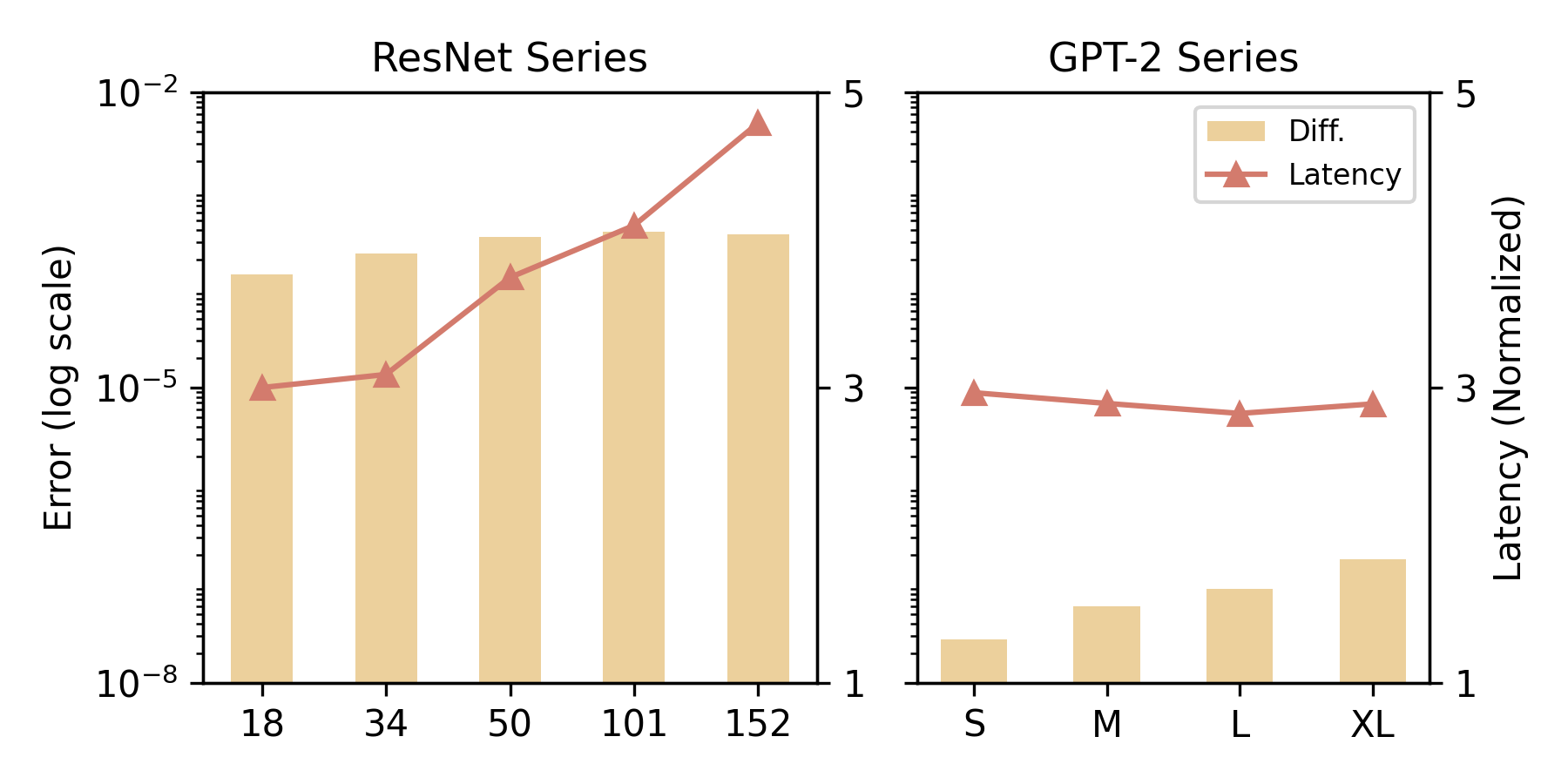} 
    \caption{Impact of model scale on accuracy and latency.}
    \label{fig:resgpt}
\end{figure}

\subsection{Inference Accuracy}
In this subsection, we evaluate the inference accuracy of \MyScheme{}. As analyzed in Section 5, our obfuscated model should theoretically produce identical computational results to the original model. The results in Table~\ref{tab:accuracy} show that the top-1 classification accuracy remains unchanged across all tested models. However, in practice, due to limited floating-point precision in matrix operations, errors are inevitably introduced. For small image recognition models, the maximum difference in final linear layer outputs between obfuscated and plaintext models remains on the order of 1e-4, with negligible impact on overall accuracy. With the increase of the model size, this error will accumulate to a non-negligible level. For instance, the error reaches 0.017 in GPT-2-xl.

To mitigate this issue, we further analyze the error magnitude introduced by different layers of the model. As shown in Table~\ref{tab:diff}, after passing through normalization layers, the accumulated errors from previous layers are significantly reduced. This is because the computational error is at least several orders of magnitude smaller than the calculated results, so after normalizing the results, the error becomes negligible. Therefore, we propose adding an additional normalization layer after the final layer of a model. Since the final layer typically produces a probability distribution over vocabulary tokens or classification labels, this modification does not affect the final output.
In Fig.~\ref{fig:resgpt}, we show the accuracy errors on GPT-2 family models after applying this optimization. The error remains below 1e-6 across all model variants, representing a reduction of several orders of magnitude. In addition, this optimization also has a negligible impact on computational overhead.

We also compare the accuracy of our method with two baselines in Table~\ref{tab:bigacc}. Since they do not provide results on LLMs, we only compare the accuracy on image classification models. TEESlice, which relies on model retraining, occasionally shows slight improvements, while TSQP exhibits minor accuracy drops due to its use of parameter projection techniques. In contrast, since \MyScheme{} does not change the structure or weights of the model, it has minimal impact on accuracy compared to similar approaches.

\subsection{Preprocessing Overhead}
Due to the ability to reuse a part of masks for linear layers, the preprocessing overhead of \MyScheme{} falls into two categories: obfuscating the entire model versus generating materials only required for non-linear layers and input processing.

During the first inference, full model obfuscation is required, with preprocessing times of: 0.471s for AlexNet, 0.096s for VGG16\_bn, 0.069s for ResNet18, 4.248s for BERT-base, and 4.475s for GPT-2-small. 
For subsequent inferences, when only generating OTPs for non-linear layers and input processing, the time overhead decreases significantly to: 0.005s for AlexNet, 0.008s for VGG16\_bn, 0.006s for ResNet18, 0.709s for BERT-base, and 0.723s for GPT-2-small. 
These results, ranging from seconds to minutes, demonstrate that \MyScheme{} introduces acceptable overhead. 
The materials generated during brief periods when devices are idle can support a substantial number of inference requests, making the approach practical for real-world deployment scenarios.

\begin{table}
\centering
    \caption{Impact of model scale on storage overhead.}
    \label{tab:storage}
    \resizebox{0.9\linewidth}{!}{ 
        \begin{tabular}{l|ccc} 
            \toprule 
            Model & Org. (MB) & Obf. (MB) & Overhead (\%) \\ 
            \midrule 
            ResNet-18  & 42.73  & 56.39  & 31.9  \\ 
            ResNet-34  & 81.35  & 103.54 & 27.3  \\ 
            ResNet-50  & 90.05  & 265.54 & 194.9 \\ 
            ResNet-101 & 162.81 & 491.01 & 201.6 \\ 
            ResNet-152 & 222.75 & 676.70 & 203.9 \\ 
            \bottomrule
        \end{tabular}
    }
\end{table}

\subsection{Storage Overhead}
Regarding the storage cost of the obfuscated model, the results in Table~\ref{tab:storage} show that as the original model scale increases, the size of its obfuscated version grows at an accelerating rate. This phenomenon occurs because larger models incorporate more non-linear layers, and the computational overhead for obfuscating non-linear layers substantially exceeds that of linear layers.
Specifically, obfuscating a linear layer incurs only limited parameter expansion, whereas a non-linear layer, which originally contains no parameters, is transformed into several large matrices, as detailed in Table~\ref{tab:ops}. As a result, the storage overhead from non-linear layers dominates the overall expansion, causing the obfuscated model size to scale with the number of such layers. Nevertheless, since the obfuscated model is stored and executed in untrusted memory, and the obfuscation process is completed during a separate preprocessing phase, the increased model size does not constitute a bottleneck for inference latency.

\section{Conclusion}

In this paper, we proposed \MyScheme{}, an efficient TEE-shielded on-device ML model protection method based on weight obfuscation. Our approach enables inference to be performed directly in untrusted memory using the obfuscated model, requiring only two rounds of interaction between untrusted and trusted memory before and after inference, and supports hardware acceleration including GPUs. We theoretically proved that the obfuscated model does not leak information about the original weights from an information-theoretic perspective. Experimental results show that \MyScheme{} preserves near-identical accuracy with inference latency within 4.8$\times$ of unprotected GPU inference, achieving about 2.2$\times$ speedup over the state-of-the-art.





\appendix


\bibliographystyle{plain}
\bibliography{ref_new}

@article{costan2016intel,
  title={Intel SGX explained},
  author={Costan, Victor},
  journal={IACR Cryptol, EPrint Arch},
  year={2016}
}

@inproceedings{zhang2024no,
  title={No privacy left outside: On the (in-) security of tee-shielded dnn partition for on-device ml},
  author={Zhang, Ziqi and Gong, Chen and Cai, Yifeng and Yuan, Yuanyuan and Liu, Bingyan and Li, Ding and Guo, Yao and Chen, Xiangqun},
  booktitle={Proc. of IEEE S\&P'24},
  pages={3327--3345},
  year={2024},
}

@inproceedings{sun2024tsqp,
  title={TSQP: Safeguarding real-time inference for quantization neural networks on edge devices},
  author={Sun, Yu and Xiong, Gaojian and Liu, Jianhua and Liu, Zheng and Cui, Jian},
  booktitle={Proc. of IEEE S\&P'25},
  pages={1--1},
  year={2025},
}

@inproceedings{sun2023shadownet,
  title={Shadownet: A secure and efficient on-device model inference system for convolutional neural networks},
  author={Sun, Zhichuang and Sun, Ruimin and Liu, Changming and Chowdhury, Amrita Roy and Lu, Long and Jha, Somesh},
  booktitle={Proc. of IEEE S\&P'23},
  pages={1596--1612},
  year={2023},
}

@inproceedings{shen2022soter,
  title={SOTER: Guarding black-box inference for general neural networks at the edge},
  author={Shen, Tianxiang and Qi, Ji and Jiang, Jianyu and Wang, Xian and Wen, Siyuan and Chen, Xusheng and Zhao, Shixiong and Wang, Sen and Chen, Li and Luo, Xiapu and others},
  booktitle={Proc. of USENIX ATC'22},
  pages={723--738},
  year={2022}
}

@inproceedings{liu2024tbnet,
  title={TBNet: A neural architectural defense framework facilitating DNN model protection in trusted execution environments},
  author={Liu, Ziyu and Zhou, Tong and Luo, Yukui and Xu, Xiaolin},
  booktitle={Proc. of DAC'24},
  pages={1--6},
  year={2024}
}

@inproceedings{liu2023mirrornet,
  title={MirrorNet: A TEE-friendly framework for secure on-device DNN inference},
  author={Liu, Ziyu and Luo, Yukui and Duan, Shijin and Zhou, Tong and Xu, Xiaolin},
  booktitle={Proc. of ICCAD'23},
  pages={1--9},
  year={2023},
}

@inproceedings{tramer2018slalom,
    title={Slalom: Fast, verifiable and private execution of neural networks in trusted hardware},
    author={Florian Tramer and Dan Boneh},
    booktitle={Proc. of ICLR'19},
    year={2019}
}

@inproceedings{mo2020darknetz,
  title={Darknetz: towards model privacy at the edge using trusted execution environments},
  author={Mo, Fan and Shamsabadi, Ali Shahin and Katevas, Kleomenis and Demetriou, Soteris and Leontiadis, Ilias and Cavallaro, Andrea and Haddadi, Hamed},
  booktitle={Proc. of MobiSys'20},
  pages={161--174},
  year={2020}
}

@inproceedings{hanzlik2021mlcapsule,
  title={Mlcapsule: Guarded offline deployment of machine learning as a service},
  author={Hanzlik, Lucjan and Zhang, Yang and Grosse, Kathrin and Salem, Ahmed and Augustin, Maximilian and Backes, Michael and Fritz, Mario},
  booktitle={Proc. of CVPR Workshop'21},
  pages={3300--3309},
  year={2021}
}

@inproceedings{elgamal2020serdab,
  title={Serdab: An IoT framework for partitioning neural networks computation across multiple enclaves},
  author={Elgamal, Tarek and Nahrstedt, Klara},
  booktitle={Proc. of CCGRID'20},
  pages={519--528},
  year={2020},
}

@inproceedings{xiang2021aegisdnn,
  title={Aegisdnn: Dependable and timely execution of dnn tasks with sgx},
  author={Xiang, Yecheng and Wang, Yidi and Choi, Hyunjong and Karimi, Mohsen and Kim, Hyoseung},
  booktitle={Proc. of RTSS'21},
  pages={68--81},
  year={2021}
}

@inproceedings{xu2024permutation,
  title={Permutation equivariance of transformers and its applications},
  author={Xu, Hengyuan and Xiang, Liyao and Ye, Hangyu and Yao, Dixi and Chu, Pengzhi and Li, Baochun},
  booktitle={Proc. of CVPR'24},
  pages={5987--5996},
  year={2024}
}

@misc{emily2023h100,
  author={Emily, Apsey and Phil, Rogers and Michael, O'Connor and Rob, Nertney},
  title={Confidential computing on NVIDIA H100 GPUs for secure and trustworthy AI},
  year={2023},
  url={https://developer.nvidia.com/blog/confidential-computing-on-h100-gpus-for-secure-and-trustworthy-ai/},
}

@inproceedings{zhanggroupcover,
  title={GroupCover: A secure, efficient and scalable inference framework for on-device model protection based on TEEs},
  author={Zhang, Zheng and Wang, Na and Zhang, Ziqi and Zhang, Yao and Zhang, Tianyi and Liu, Jianwei and Wu, Ye},
  booktitle={Proc. of ICML'24},
  pages={59992--60003},
  year={2024}
}

@article{sev2020strengthening,
  title={Strengthening VM isolation with integrity protection and more},
  author={Sev-Snp, AMD},
  journal={White Paper, January},
  volume={53},
  pages={1450--1465},
  year={2020}
}

@article{pinto2019demystifying,
  title={Demystifying arm trustzone: A comprehensive survey},
  author={Pinto, Sandro and Santos, Nuno},
  journal={ACM computing surveys (CSUR)},
  pages={1--36},
  year={2019},
}

@inproceedings{vaswani2017attention,
  title={Attention is all you need},
  author={Vaswani, Ashish and Shazeer, Noam and Parmar, Niki and Uszkoreit, Jakob and Jones, Llion and Gomez, Aidan N and Kaiser, {\L}ukasz and Polosukhin, Illia},
  booktitle={Proc. of NeurIPS'17},
  year={2017}
}

@inproceedings{wangglue,
  title={GLUE: A multi-task benchmark and analysis platform for natural language understanding},
  author={Wang, Alex and Singh, Amanpreet and Michael, Julian and Hill, Felix and Levy, Omer and Bowman, Samuel R},
  booktitle={Proc. of ICLR'18},
  year={2018}
}

@inproceedings{he2016deep,
  title={Deep residual learning for image recognition},
  author={He, Kaiming and Zhang, Xiangyu and Ren, Shaoqing and Sun, Jian},
  booktitle={Proc. of CVPR'16},
  pages={770--778},
  year={2016}
}

@inproceedings{simonyan2015very,
  title={Very deep convolutional networks for large-scale image recognition},
  author={Simonyan, K and Zisserman, A},
  booktitle={Proc. of ICLR'15},
  year={2015}
}

@inproceedings{krizhevsky2012imagenet,
  title={Imagenet classification with deep convolutional neural networks},
  author={Krizhevsky, Alex and Sutskever, Ilya and Hinton, Geoffrey E},
  booktitle={Proc. of NeurIPS'12},
  year={2012}
}

@inproceedings{devlin2019bert,
  title={Bert: Pre-training of deep bidirectional transformers for language understanding},
  author={Devlin, Jacob and Chang, Ming-Wei and Lee, Kenton and Toutanova, Kristina},
  booktitle={Proc. of NAACL'19},
  pages={4171--4186},
  year={2019}
}

@misc{krizhevsky2009learning,
  title={Learning multiple layers of features from tiny images.},
  author={Krizhevsky, Alex and Hinton, Geoffrey and others},
  year={2009}
}

@Misc{appleTEE,
	note = {\url{https://support.apple.com/en-sg/guide/security/sec59b0b31ff/web}},
	title = {Secure Enclave - Apple Support (SG).},
    author = {Apple},
	year = {2024}
}

@Misc{iotTEE,
	note = {\url{https://www.intel.com/content/www/us/en/products/details/embedded-processors.html}},
	title = {Product brief, 3rd Gen Intel Xeon scalable processor for IoT.},
    author = {Intel},
	year = {2023}
}

@misc{inteltdx,
    note = {\url{https://cdrdv2.intel.com/v1/dl/getContent/690419}},
    title = {Intel Trust Domain Extensions.},
    author = {Intel},
    year = {2023}
}

@article{carlini2020cryptanalytic,
  title={Cryptanalytic extraction of neural network models},
  author={Carlini, Nicholas and Jagielski, Matthew and Mironov, Ilya},
  journal={Advances in Cryptology — Crypto'20},
  pages={189--218},
  year={2020}
}

@article{radford2019language,
  title={Language models are unsupervised multitask learners},
  author={Radford, Alec and Wu, Jeffrey and Child, Rewon and Luan, David and Amodei, Dario and Sutskever, Ilya and others},
  journal={OpenAI blog},
  volume={1},
  pages={9},
  year={2019}
}

@inproceedings{sun2021mind,
  title={Mind your weight (s): A large-scale study on insufficient machine learning model protection in mobile apps},
  author={Sun, Zhichuang and Sun, Ruimin and Lu, Long and Mislove, Alan},
  booktitle={Proc. of USENIX Security'21},
  pages={1955--1972},
  year={2021}
}

@inproceedings{orekondy2019knockoff,
  title={Knockoff nets: Stealing functionality of black-box models},
  author={Orekondy, Tribhuvanesh and Schiele, Bernt and Fritz, Mario},
  booktitle={Proc. of CVPR'19},
  pages={4954--4963},
  year={2019}
}

@inproceedings{yuan2024hypertheft,
  title={Hypertheft: Thieving model weights from tee-shielded neural networks via ciphertext side channels},
  author={Yuan, Yuanyuan and Liu, Zhibo and Deng, Sen and Chen, Yanzuo and Wang, Shuai and Zhang, Yinqian and Su, Zhendong},
  booktitle={Proc. of CCS'24},
  pages={4346--4360},
  year={2024}
}

@book{lindell2017tutorials,
  title={Tutorials on the Foundations of Cryptography},
  author={Lindell, Yehuda},
  year={2017},
  publisher={Springer}
}

@inproceedings{wang2025arrow,
  title={Game of Arrows: On the (In-) Security of Weight Obfuscation for On-Device TEE-Shielded LLM Partition Algorithms},
  author={Wang, Pengli and Dong, Bingyou and Cai, Yifeng and Zhang, Zheng and Liu, Junlin and Xue, Huanran and Wu, Ye and Zhang, Yao and Zhang, Ziqi},
  booktitle={Proc. of USENIX Security'25},
  year={2025}
}

@inproceedings{nayan2024sok,
  title={SoK: All you need to know about $\{$On-Device$\}$$\{$ML$\}$ model extraction-the gap between research and practice},
  author={Nayan, Tushar and Guo, Qiming and Al Duniawi, Mohammed and Botacin, Marcus and Uluagac, Selcuk and Sun, Ruimin},
  booktitle={Proc. of USENIX Security'24)},
  pages={5233--5250},
  year={2024}
}

@article{li2024personal,
  title={Personal llm agents: Insights and survey about the capability, efficiency and security},
  author={Li, Yuanchun and Wen, Hao and Wang, Weijun and Li, Xiangyu and Yuan, Yizhen and Liu, Guohong and Liu, Jiacheng and Xu, Wenxing and Wang, Xiang and Sun, Yi and others},
  journal={arXiv preprint arXiv:2401.05459},
  year={2024}
}

@article{zhao2022survey,
  title={A survey of deep learning on mobile devices: Applications, optimizations, challenges, and research opportunities},
  author={Zhao, Tianming and Xie, Yucheng and Wang, Yan and Cheng, Jerry and Guo, Xiaonan and Hu, Bin and Chen, Yingying},
  journal={Proc. of the IEEE},
  pages={334--354},
  year={2022}
}

@article{russakovsky2015imagenetdataset,
  title={Imagenet large scale visual recognition challenge},
  author={Russakovsky, Olga and Deng, Jia and Su, Hao and Krause, Jonathan and Satheesh, Sanjeev and Ma, Sean and Huang, Zhiheng and Karpathy, Andrej and Khosla, Aditya and Bernstein, Michael and others},
  journal={International journal of computer vision},
  pages={211--252},
  year={2015},
}

@inproceedings{duy2025incognitos,
  title={INCOGNITOS: A Practical Unikernel Design for Full-System Obfuscation in Confidential Virtual Machines},
  author={Duy, Kha Dinh and Kim, Jaeyoon and Lim, Hajeong and Lee, Hojoon},
  booktitle={2025 IEEE Symposium on Security and Privacy (SP)},
  pages={4192--4209},
  year={2025},
  organization={IEEE}
}

@inproceedings{wichelmann2024obelix,
  title={Obelix: Mitigating side-channels through dynamic obfuscation},
  author={Wichelmann, Jan and Rabich, Anja and P{\"a}tschke, Anna and Eisenbarth, Thomas},
  booktitle={2024 IEEE Symposium on Security and Privacy (SP)},
  pages={4182--4199},
  year={2024},
  organization={IEEE}
}

\section{Security Analysis}
\label{sec.ana}
\subsection{Proof of Theorem 5.1}

\begin{proof}
In particular, as illustrated in Alg.~\ref{alg:allprocess}, \MyScheme{} uses three uniformly random and invertible matrices $P, Q$ and $S$, all of which are sampled independently and remain secret to the adversary.
In each round $i \in \{1, 2, \dots, t\}$, where \(t\) is polynomially bounded, 
the system selects a fresh random matrix \(T_i\).
Without loss of generality, we set \(X_i = 0\) to streamline the subsequent security analysis.

The adversary then observes the tuple
\[
\tilde{\mathcal{O}_i} \;=\; \bigl(\tilde{X}_i, \tilde{T}_i, \tilde{W}\bigr)
\;=\; \bigl(-P\,T_i\,Q,\; P\,T_i\,W\,S,\; Q^{-1}\,W\,S\bigr)
\]
in each round.
Consequently, demonstrating that the adversary gains no information about \(W\) 
is equivalent to showing that the following mutual information equals zero:
\[
I\bigl(W \; ; \;\tilde{\mathcal{O}_1}, \dots, \tilde{\mathcal{O}_t}\bigr) 
\;=\; 0.
\]

We begin by analyzing the impact of observing $\tilde{W} = Q^{-1} W S$. Since $Q$ and $S$ are independently and uniformly sampled invertible matrices, the set of possible $W$ matrices that could yield a given $\tilde{W}$ spans the entire space. Indeed, for any particular $\tilde{W}$, there exist invertible matrices $Q', S'$ such that $W = Q' \,\tilde{W}\, (S')^{-1}$. Hence, from the adversary's perspective, conditioning on $\tilde{W}$ does not reduce the uncertainty of $W$, implying
\[
H(W \mid \tilde{W}) = H(W).
\]

Next, we consider $\tilde{X}_i = -P T_i Q$. Here, $T_i$ is freshly and uniformly chosen in each round, while $P, Q$ remain unknown and uniformly sampled from the set of invertible matrices. Consequently, the adversary cannot deduce $T_i$ or $P T_i$ from $\tilde{X}_i$, and so
\[
H(W \mid \tilde{X}_i, \tilde{W}) = H(W \mid \tilde{W}).
\]

Further, recall that $\tilde{T}_i = P T_i W S$. Since $P T_i$ is uniformly random and independent of $W$, and $S$ is likewise unknown, $\tilde{T}_i$ provides no additional information about $W$. In particular, for any two candidate matrices $W_1, W_2$, one can find a suitable $T_i'$ so that $P T_i' W_1 S = P T_i W_2 S$. Thus,
\[
H(W \mid \tilde{X}_i, \tilde{T}_i, \tilde{W}) = H(W \mid \tilde{\mathcal{O}_i}) = H(W \mid \tilde{W}).
\]

As the fresh randomness $T_i$ is used in each round independently, the same argument applies to every round. By the chain rule of entropy,
\[
H(W \mid \tilde{\mathcal{O}}_1, \dots, \tilde{\mathcal{O}}_t)
= H(W \mid \tilde{W}).
\]
Combining these observations, we conclude
\[
H(W \mid \tilde{\mathcal{O}}_1, \dots, \tilde{\mathcal{O}}_t) = H(W),
\]
and therefore the mutual information satisfies
\[
I(W ; \tilde{\mathcal{O}}_1, \dots, \tilde{\mathcal{O}}_t)
= H(W) - H(W \mid \tilde{\mathcal{O}}_1, \dots, \tilde{\mathcal{O}}_t)
= 0.
\]
This completes the proof that the adversary learns no information about $W$ in an information-theoretic sense.
\end{proof}

\subsection{Proof of Theorem 5.2}

\begin{proof}
The non-linear layer employs two independently sampled, uniformly random, and adversary-secret invertible matrices $P$ and $Q$, as well as per-inference independently and uniformly sampled random matrices $\{R_j\}_{j=1}^3$ and permutations $\{\pi_j\}_{j=1}^4$.

We first analyze the observation $\tilde{X} = PXQ$. Since both $P$ and $Q$ are independently and uniformly sampled invertible matrices, for any given $\tilde{X}$, the set of possible $X$ matrices spans the entire space. Specifically, for any $\tilde{X}$, there exist invertible matrices $P'$ and $Q'$ such that $X = P' \tilde{X} Q'$. Therefore, from the adversary's perspective, $\tilde{X}$ does not reduce the uncertainty of $X$, yielding:
\[
H(X \mid \tilde{X}) = H(X).
\]

Next, considering $M_1 = \pi_3(\pi_1P^{-1}\otimes R_1)$, since $\pi_3$, $\pi_1$, and $R_1$ are all freshly and uniformly sampled random matrices in each round, while $P$ remains unknown and uniformly sampled, the adversary cannot derive any information about $X$ from $M_1$. The same holds for $M_2$, $M_3$, $M_4$, and $R_2$. Thus, we have:
\[
H(X \mid \tilde{X}, M_1, M_2, M_3, M_4, R_2) = H(X \mid \tilde{X}).
\]

Consequently, the mutual information satisfies:
\[
I(X; \mathcal{O}) = H(X) - H(X \mid \tilde{X}, M_1, M_2, M_3, M_4, R_2) = 0.
\]

This completes our proof that the adversary gains no information about the correct inputs or outputs during the computation of the non-linear layer.
\end{proof}

\subsection{Proof of Theorem 5.3}

\begin{proof}
The intermediate linear layer employs three independently sampled, uniformly random, and adversary-secret invertible matrices $P$, $Q$, and $S$. Across polynomially bounded rounds $i \in \{1, 2, \cdots, t\}$, the adversary observes the following tuple in each observation:
\[
\mathcal{O}_i = \bigl(\tilde{X}_i, \tilde{W}\bigr) = \bigl(PX_iQ,\; Q^{-1}WS\bigr)
\]

We first analyze the effect of observing $\tilde{W} = Q^{-1}WS$. Since $Q$ and $S$ are independently and uniformly sampled invertible matrices, for any given $\tilde{W}$, the set of possible $W$ matrices spans the entire space. Specifically, for any particular $\tilde{W}$, there exist invertible matrices $Q'$ and $S'$ such that $W = Q'\tilde{W}(S')^{-1}$. Hence, from the adversary's perspective, conditioning on $\tilde{W}$ does not reduce the uncertainty about $W$, yielding:
\[
H(W \mid \tilde{W}) = H(W).
\]

Next, we examine $\tilde{X}_i = PX_iQ$. The input $X_i$ of an intermediate linear layer should be the output of the prior non-linear layer. According to Theorem \ref{theorem:2}, the adversary gains no information about $X_i$, and since $P$ and $Q$ remain unknown and uniformly sampled, the adversary cannot deduce $X_i$ or $PX_i$ from $\tilde{X}_i$. Therefore:
\[
H(W \mid \tilde{X}_i, \tilde{W}) = H(W \mid \mathcal{O}_i) = H(W \mid \tilde{W}).
\]
As the value of $X_i$  is unknown to the adversary in all rounds, this argument holds universally across all observations. Thus we obtain:
$$
H(W \mid \mathcal{O}_1, \cdots, \mathcal{O}_t) = H(W \mid \tilde{W}).
$$
Combining these results, we conclude:
$$
H(W \mid \mathcal{O}_1, \cdots, \mathcal{O}_t) = H(W),
$$
and consequently, the mutual information satisfies:
\[
I(W ; \mathcal{O}_1, \dots, \mathcal{O}_t) = H(W) - H(W \mid \mathcal{O}_1, \dots, \mathcal{O}_t) = 0.
\]

This completes the proof that the adversary gains zero information about $W$ in the information-theoretic sense.
\end{proof}

\section{Error Analysis}

\begin{table}[htbp]
\caption{Detailed Error Analysis for a Transformer Block.}
\label{tab:diff}
\centering
\begin{tabular}{lc}
\toprule
Component & Error Value (diff.) \\
\midrule
Self-attention & 5.2947 \\
LayerNorm (1st) & 0.005 \\
Linear (1st) & 0.1397 \\
Activation & 0.721 \\
Linear (2nd) & 5.4967 \\
LayerNorm (2nd) & 0.0049 \\
\bottomrule
\end{tabular}
\end{table}

\end{document}